\documentclass[USenglish,a4paper,11pt]{article}
\pdfoutput=1

\input{Header/style.sty}
\title{Dequantization Barriers for Guided Stoquastic Hamiltonians}

\author{Yassine Hamoudi\thanks{\texttt{ys.hamoudi@gmail.com}}\qquad\; Yvan Le Borgne\thanks{\texttt{borgne@labri.fr}}\qquad\; Shrinidhi Teganahally Sridhara\thanks{\texttt{shrinidhi.teganahally-sridhara@u-bordeaux.fr}} \\ [12pt] Université de Bordeaux, CNRS, LaBRI, France}

\date{\today}

\begin{document}

\maketitle

\begin{abstract}
  
We construct a probability distribution, induced by the Perron--Frobenius eigenvector of an exponentially large graph, that cannot be efficiently sampled by any classical algorithm, even when given the \emph{best-possible} warm-start distribution. In the quantum setting, this problem can be viewed as preparing the ground state of a stoquastic Hamiltonian given a guiding state as input, and is known to be efficiently solvable on a quantum computer. We further show that even estimating the ground-state energy of that Hamiltonian remains classically hard. Our results suggest that no efficient classical algorithm can solve a broad class of stoquastic ground-state problems.

Our graph is constructed from a class of \emph{high-degree, high-girth spectral expanders} (such as the LPS Ramanujan graphs) to which self-similar trees are attached. This builds on and extends prior work of Gilyén, Hastings, and Vazirani~\cite{Has21j,GHV21c}, which ruled out dequantization for a specific stoquastic adiabatic path algorithm. We strengthen their result by ruling out any classical algorithm for guided ground-state preparation or energy estimation,  while also providing a derandomized construction.
\end{abstract}


\section{Introduction}
Understanding and analyzing the ground state~$\psh$ of a many-body system represented by a Hamiltonian~$H$ acting on~$m$ qubits is of fundamental interest in quantum computation with critical applications in quantum chemistry (molecular modeling), materials science (high-temperature superconductivity), and quantum machine learning, where ground states are used to represent optimal model parameters. One of the main obstructions to efficient preparation of the ground state or estimation of the ground-state energy is the exponential size of the matrix~$H$. In fact, in the absence of other information apart from the Hamiltonian, the problem of estimating the ground-state energy is computationally hard even for quantum computers; the problem is famously QMA-complete for local Hamiltonians~\cite{KSV02b}. In practice, we often assume access to \emph{guiding states}~$\psin$, which are states close to the actual ground state (measured by the overlap~$\ip{\psi_{\mathrm{in}}}{\psi_H}$), as a computational resource for the ground-state preparation/energy estimation problems. We refer to this setting as the \emph{Guided Ground-State Preparation} (GGSP) problem and to the related energy estimation problem as \emph{Guided Ground-State Energy Estimation} (GGSE). 

Guiding states can be seen as analogous to \emph{warm starts} in Markov Chain Monte Carlo algorithms \cite{MiscAF02}, ansätze in variational quantum algorithms \cite{TCC+22j}, or the intermediate states along a path of adiabatic quantum computation~\cite{AL18j}. Guiding states can also be classified based on their closeness to the actual ground state and the existence of efficient algorithms to prepare them or query their amplitudes. The well-studied types include guiding states with constant or~$1/\poly(m)$ \emph{global} overlap, \emph{strong} guiding states with constant or~$1/\poly(m)$ \emph{pointwise} overlap \cite{Bra15j}, \emph{succinct} guiding states with efficient classical circuits for amplitude evaluation \cite{Liu21c,GLG22c,Jia25j,LGal25c}, etc. In the presence of these guiding states, the GGSP/GGSE problems have been proved to be efficiently solvable by quantum algorithms and, in some cases, even by classical algorithms under additional assumptions on~$H$ and~$\psin$. 

\paragraph{Dequantization.}
Recently, there have been several results on \emph{dequantized} algorithms \cite{GLG22c, LGal25c,CFG+23c} for the GGSE problem. These \emph{classical} algorithms usually assume sampling access to the guiding state, along with amplitude access to the guiding state via an evaluation oracle. For instance, Le Gall \cite{LGal25c} showed that when~$\abs{\ip{\psi_{\mathrm{in}}}{\psi_H}} = \om{1}$ and the required precision is~$1/\log(m)$, there exists a classical algorithm for energy estimation that runs in~$\poly(m)$ time. In fact, when the exact ground state is given as input, $\ket{\psi_{\mathrm{in}}} = \ket{\psi_H}$, a straightforward algorithm computes the energy \emph{exactly}, provided amplitude access is available (see, e.g., the comment in~\cite{GLG22c}). Our result shows that sampling access alone is insufficient, even when the guiding state is very close to the exact ground state, thereby demonstrating the necessity of amplitude access for efficient classical energy estimation in this setting.

Another line of classical algorithms for the GGSP problems is the Quantum Monte Carlo (QMC) method, a widely used heuristic in computational condensed matter physics. Bravyi and Terhal~\cite{BT10j} showed that when~$H$ is stoquastic and frustration-free, then the ground-state preparation is classically tractable. Later, in \cite{BCGL23j}, it was shown that there exists a rapidly mixing Markov Chain that can sample from the ground state of any gapped, local~$H$ when given access to ratios of amplitudes of the ground state via an oracle, along with a starting state satisfying a mild technical condition. However, implementing the amplitude-ratio oracle is unknown for general Hamiltonians, including stoquastic Hamiltonians that are not frustration-free.

\paragraph{General barriers.} Dequantization results have led to the question of \emph{to what extent classical algorithms can suffice to solve the GGSP/GGSE problems}. This question has received most of its attention when the guiding state is an \emph{untrusted} witness, i.e., claimed by the prover in prover-verifier protocols. While the general energy estimation problem in this setting is famously QMA-complete~\cite{KSV02b,GHLS14j}, restricting to Hamiltonians that admit specific guiding states makes it QCMA-complete~\cite{WFC24c} or MA-complete~\cite{Jia25j}. Turning to the \emph{trusted} case, where a guiding state is given as input (as in the present work), GGSE is BQP-complete when the required precision is~$1/\poly(m)$ and the guiding state is chosen adversarially~\cite{CFG+23c,GLG22c}. For weaker precision, the problem becomes solvable in BPP~\cite{GLG22c,LGal25c}.

\paragraph{Stoquastic barriers.}
A further restriction on Hamiltonians is to assume them to be \emph{stoquastic}, namely, that all their off-diagonal entries are real and non-positive in some given basis. Such Hamiltonians are particularly interesting because of their classical flavor. Indeed, the ground state of a stoquastic Hamiltonian can be chosen to have non-negative real entries, which allows it to be interpreted as a probability distribution. Most Hamiltonians arising in quantum architectures like D-Wave come from Transverse Field Ising Models (TIM), which are stoquastic in the standard basis. Other examples of stoquastic Hamiltonians include the Heisenberg model on bipartite graphs and the bosonic Hubbard model. Since a classical objective function can be seen as a diagonal Hamiltonian and hence stoquastic, the ground-state properties of stoquastic Hamiltonians are of considerable importance in optimization. The general energy estimation problem becomes StoqMA-complete when~$H$ is stoquastic, with $\text{MA} \subseteq \text{StoqMA} \subseteq \text{SBP} \cap \text{QMA}$ \cite{BDOT08j,BBT06p}. In the untrusted setting, the problem collapses to MA-complete under various assumptions about the guiding states~\cite{Bra15j,Liu21c}. In the trusted setting, the problem is only known to be BPP-hard~\cite{Wai25p}. We note that some of these results, such as the AM and SBP containments \cite{BDOT08j,BBT06p}, or the recent StoqMA-completeness result of \cite{GR26p}, do not require the Hamiltonian to be local, but only sparse (as is the case for the Hamiltonian constructed in our work).

\paragraph{Connection to Adiabatic Quantum Computation.}
Adiabatic quantum computation is a leading framework for solving ground-state preparation problems on a quantum computer when a target Hamiltonian has no available guiding state initially \cite{AT07j}. The paradigm relies on designing a sequence (or \emph{path}) of slowly varying Hamiltonians whose spectral gap remains open, ensuring -- by the adiabatic theorem -- that the system stays close to its instantaneous ground state throughout the evolution. From a computational perspective, this process can be viewed as a \emph{bootstrapped} guided ground-state preparation algorithm, where access to intermediate ground states enables the preparation of the final target state. The possibility of dequantizing adiabatic algorithms has been raised in several prior works. For instance, \cite{BDOT08j} asked whether \emph{``every efficient adiabatic path using stoquastic Hamiltonians can be simulated by a polynomial-time probabilistic machine''}. Positive results in this direction include a BPP algorithm for simulating any path of stoquastic frustration-free Hamiltonians \cite{BT10j}, and a PostBPP algorithm for simulating paths of general stoquastic Hamiltonians \cite{BDOT08j}. On the negative side, recent breakthrough results rule out the possibility of dequantizing a \emph{specific} stoquastic path in BPP, relative to an oracle \cite{Has21j,GHV21c}. Our results align with these latter works but have broader implications: we rule out any dequantization of an adiabatic path -- stoquastic or not -- that must pass through the stoquastic Hamiltonian constructed in our paper.

\paragraph{Connection to QMA vs QCMA separation oracles.}

A long-standing open question in quantum complexity theory asks whether quantum witnesses are more powerful than classical witnesses in quantum verification protocols. Several works have attacked this question via oracle separations~\cite{AK07j}, culminating recently in positive results relative to oracles inspired by Forrelation~\cite{BHNZ26c} and the Yamakawa-Zhandry problem~\cite{BHV26p}. Some prior candidates for such separations are based on expander mixing problems with graph oracles~\cite{Lut11p,NN24j,BFM23c,LMY25c}. A key difference with our construction is that the graphs are degree-regular, so the ground energy (and ground state) of their negative adjacency matrices are trivial to compute. These works instead rely on testing the second-lowest eigenvalue. In general, there is presumably a gap in hardness between computing the ground and excited-state energies of stoquastic Hamiltonians, with the latter being QMA-complete rather than StoqMA-complete~\cite{JGL10j}. In contrast to these works, our construction poses a genuine ground-state problem, and may lead to a tighter oracle separation between StoqMA and QCMA (or MA).

\paragraph{Connection to other exponential speedups.}
Our contribution can also be viewed as part of a line of work showcasing exponential quantum speedups based on oracular graph constructions with particular spectral properties. A foundational result in this direction is due to Childs et al.~\cite{CCD03c} for the so-called \emph{welded-tree problem}. These ideas have since been extended to demonstrate exponential quantum speedups in a variety of settings, including quantum-walk-based search algorithms~\cite{CCD03c,LLL24j,BLH25j,JZ23c,LZ25j}, quantum property testing~\cite{BCG+24j}, quantum annealing~\cite{SNK12j}, adiabatic computing~\cite{Has21j,GHV21c}, quantum distributed computing~\cite{GLMR25p}, and quantum optimization~\cite{LWWZ25p}.


\subsection{Contributions}

Our first main result is a no-go theorem for the efficient dequantization of GGSP for stoquastic Hamiltonians, even when the algorithm has access to an arbitrarily good guiding state. Below, we give an informal definition of the problem.

\begin{rdefinition}[\Cref{Def:GGSP} (Informal definition of GGSP)]
  We are given an~$m$-qubit stoquastic Hamiltonian~$H$ with a unique \emph{ground state} $\psh$ and a spectral gap of at least $1/\poly(m)$. In addition, we are given independent copies of a \emph{guiding state} $\ket{\psi_{\mathrm{in}}}$ satisfying $\abs{\ip{\psi_{\mathrm{in}}}{\psi_H}} \geq 1/\poly(m)$. The goal is to output a \emph{refined state} $\ket{\psi_{\mathrm{out}}}$ such that $\abs{\ip{\psi_{\mathrm{out}}}{\psi_H}} \geq 3/4$.

  For classical algorithms, the input states are replaced with i.i.d. samples from the distribution~$(\ip{x}{\psi_{\mathrm{in}}}^2)_x$. The output must be a sample from the distribution $(\ip{y}{\psi_{\mathrm{out}}}^2)_y$ and it must be approximately statistically independent of the input samples.
\end{rdefinition}

Note that we allow guiding states that need not be easy to prepare, including guiding states that are exactly equal to the ground state. This generality enables us to understand the power of a wide range of guiding-state preparation procedures, such as more general adiabatic paths than those previously considered (e.g., in~\cite{Has21j,GHV21c}) and arbitrary ansätze used as guiding states.

We focus on GGSP because the ability to sample from or prepare the ground state is useful for a variety of tasks beyond energy estimation, such as estimating expectation values of general observables on the ground state, or simulating its real-time dynamics. It also helps decouple the complexity of these problems from the hardness of preparing good guiding states. For instance, the dequantized energy estimation algorithm of~\cite{LGal25c} runs in $\poly\pt[\big]{\frac{1}{\abs{\ip{\psi_{\mathrm{in}}}{\psi_H}}^{1/\Delta}},m}$ (where $\Delta$ is the desired energy precision), which could be improved given the ability to boost the overlap~$\abs{\ip{\psi_{\mathrm{in}}}{\psi_H}}$ via an independent GGSP algorithm. Our result partly rules out this possibility by establishing classical hardness of GGSP for a stoquastic Hamiltonian, even in the presence of an arbitrary guiding state. Informally, we establish this classical hardness by proving the following result.

\begin{rtheorem}[\Cref{thm:lowerbound} and \Cref{Prop:quantalg} (Informal)]
  There exists an $m$-qubit stoquastic Hamiltonian~$H$ for which the GGSP problem can be solved efficiently (in~$\poly(m)$ time) by a quantum algorithm, whereas no classical algorithm can solve it using at most~$2^{m^c}$ input samples and queries to~$H$, for some constant~$c<1$. This result holds \emph{for any} guiding state provided as input.
\end{rtheorem}

Our second main result is a similar sub-exponential lower bound for classical algorithms solving the ground-state energy estimation problem. We give an informal description of the problem below.

\begin{rdefinition}[\Cref{Def:ggee} (Informal definition of GGSE)]
  We are given an~$m$-qubit stoquastic Hamiltonian~$H$ with a \emph{ground state} $\psh$. In addition, we are given copies of a \emph{guiding state}~$\ket{\psi_{\mathrm{in}}}$ satisfying $\abs{\ip{\psi_{\mathrm{in}}}{\psi_H}} \geq 1/\poly(m)$. The input is promised to satisfy exactly one of the following: $\yes$ case: $\lambda_{\rm min}(H) \leq a$ or $\no$ case: $\lambda_{\rm min}(H) \geq b$ with $b-a \geq 1/\poly(m)$. The goal is to decide between the two cases on every input satisfying the promise. For classical algorithms, the input states are replaced with i.i.d. samples from the distribution~$(\ip{x}{\psi_{\mathrm{in}}}^2)_x$.
\end{rdefinition}

Existing dequantization results for GGSE (e.g.,~\cite{GLG22c,LGal25c}) rely on two additional assumptions: query access to the amplitudes of the guiding state and a constant precision $b-a = \om{1}$. On the other hand, these classical algorithms apply to arbitrary sparse Hamiltonians. Our result shows that restricting the Hamiltonian to be stoquastic is not enough to eliminate these assumptions. In particular, this suggests that in settings where amplitude access is not naturally available -- such as adiabatic quantum computation or certain constructions aimed at oracle separations between quantum and classical proofs -- efficient classical algorithms should not be expected for arbitrary stoquastic Hamiltonians. Informally, we prove the following.

\begin{rtheorem}[\Cref{Thm:GEEHard} and \Cref{Prop:q-algo-ggse} (Informal)]
  There exists an $m$-qubit stoquastic Hamiltonian~$H$ for which the GGSE problem can be solved efficiently (in~$\poly(m)$ time) by a quantum algorithm, whereas no classical algorithm can solve it using at most~$2^{m^c}$ input samples and queries to~$H$, for some constant~$c<1$. This result holds for a guiding state $\ket{\psi_{\mathrm{in}}}$ satisfying~$\abs{\ip{\psi_{\mathrm{in}}}{\psi_H}} \geq 1 - 1/\poly(m)$.
\end{rtheorem}

We stress that, unlike our result for GGSP, the classical lower bound for GGSE holds only for a specific choice of the guiding state. This restriction is necessary because certain guiding states can make the GGSE problem trivial to solve (for instance, by hardcoding the bit corresponding to the answer). We also note that Hastings' result~\cite{Has21j} can be interpreted as a classical hardness result for GGSE when the guiding state has inverse-polynomial overlap, $1/\poly(m)$, with the ground state. Our construction demonstrates that classical hardness persists even when the overlap is as large as $1 - 1/\poly(m)$.


\subsection{Proof overview}

\subsubsection{Guided Ground-State Preparation}

We provide a high-level overview of the ideas behind the proof of \Cref{thm:lowerbound}. Our result holds relative to an oracle representing the adjacency matrix of a graph $G$ constructed in this paper.

\paragraph{Underlying Graph.}
We construct a graph~$G$ (\Cref{Def:maingraph} and \Cref{fig:overview}) that is obtained by taking a graph~$E$ which is a high-degree, high-girth spectral expander and then attaching many copies of self-similar trees~$\treh$ to its nodes (these are similar to the trees studied in \cite{Has21j,GHV21c}). $G$ is constructed in a way that it is degree-regular everywhere except at the leaves. The expander results in~$G$ having a large spectral gap and no short cycles (by definition, the girth is the length of the shortest cycle in the graph). The stoquastic Hamiltonian~$H$ acting on~$m$ qubits is the negative of the adjacency matrix of~$G$ over~$2^m$ vertices, i.e.,~$H = -A_G$. The only assumption on the guiding state~$\psin$ is that it has at least inverse-polynomial overlap with~$\psh$ since, without this assumption, even a quantum computer may not be able to efficiently solve the GGSP problem. We show that \emph{no such guiding state} can simultaneously enable an efficient classical algorithm.

\usetikzlibrary{decorations.pathmorphing}

\newcommand{\bigpicturebisdecoratednode}[2]{
\tikz[scale=#1,rotate=#2]{
      \draw[line width=1,blue] (0,0) -- +(300:1.5);
      \draw[line width=1,blue] (0,0) -- +(320:1.5);
      \fill[line width=1,draw=blue,fill=white,rounded corners] (320:1.5) -- (300:3) -- (340:3) -- cycle;
      \fill[line width=1,draw=blue,fill=white,rounded corners] (300:1.5) -- (280:3) -- (320:3) -- cycle;
      \fill[fill=white,draw=blue] (320:1.5) circle (0.2);
      \fill[fill=white,draw=blue] (300:1.5) circle (0.2);
      }}

\newcommand{\bigpicturebisdecoratednodeindecoration}[2]{
\tikz[scale=#1,rotate=#2]{
      \draw[line width=1,blue] (0,0) -- +(300:1.5);
      \draw[line width=1,blue] (0,0) -- +(320:1.5);
      \fill[line width=1,draw=blue,fill=white,rounded corners] (320:1.5) -- (300:3) -- (340:3) -- cycle;
      \fill[line width=1,draw=blue,fill=white,rounded corners] (300:1.5) -- (280:3) -- (320:3) -- cycle;
      \fill[fill=white,draw=blue] (320:1.5) circle (0.2);
      \fill[fill=white,draw=blue] (300:1.5) circle (0.2);
      }}

\newcommand{\twotreebis}[2]{
    \foreach \xa/\ya/\xb/\yb in {0/0/-2/1,0/0/2/1, 
      -2/1/-3/2,-2/1/-1/2,2/1/1/2,2/1/3/2 
    }{
      \draw[line width=1,draw=black] (\xa,\ya) -- (\xb,\yb);
    }
    \foreach \x in {-3,-1,1,3}{
      \draw[line width=1, draw=red] (\x,2) -- +(#2:1);
      \draw[line width=1, draw=blue] (\x,2) -- +(#2+30:1);
      \draw[line width=1, draw=blue] (\x,2) -- +(#2+20:1);
      \fill[fill=white,draw=#1] (\x,2) circle (0.2);
    }
    \foreach \x in {-2,2}{
      \draw[line width=1, draw=red] (\x,1) -- +(#2:1); 
      \draw[line width=1, draw=blue] (\x,1) -- +(#2+30:1);
      \draw[line width=1, draw=blue] (\x,1) -- +(#2+20:1);
      \fill[fill=white,draw=#1] (\x,1) circle (0.2);
    }
    \draw[line width=1, draw=red] (0,0) -- +(#2:1);   
    \draw[line width=1, draw=blue] (0,0) -- +(#2+30:1);
    \draw[line width=1, draw=blue] (0,0) -- +(#2+20:1);
    \fill[fill=white,draw=#1] (0,0) circle (0.2);
}

\newcommand{\threetreebis}[1]{
    \foreach \xa/\ya/\xb/\yb in {0/0/3/1, 
      -3/1/-2/2,0/1/1/2,3/1/4/2
    }{
      \draw[line width=1,draw=#1,decorate,decoration=snake] (\xa,\ya) -- (\xb,\yb);
      }
    \foreach \xa/\ya/\xb/\yb in {0/0/-3/1,0/0/0/1, 
      -3/1/-4/2,-3/1/-3/2,0/1/-1/2,0/1/0/2,3/1/2/2,3/1/3/2
    }{
      \draw[line width=1,draw=black] (\xa,\ya) -- (\xb,\yb);
      }
    \foreach \x in {-4,...,4}{
      \fill[fill=white,draw=black] (\x-0.17,2-0.17) rectangle (\x+0.17,2+0.17);
    }
    \foreach \x in {-3,0,3}{
      \fill[fill=white,draw=black] (\x-0.17,1-0.17) rectangle (\x+0.17,1+0.17);
    }
    \fill[fill=white,draw=black] (0-0.17,0-0.17) rectangle (0+0.17,0+0.17);
}

\newcommand{\mytree}[1]{\tikz[scale=0.5]{
    \draw[rounded corners,fill=white,line width=1] (0,0) -- (-0.8,2) -- (0.8,2) -- cycle;
    \fill[draw=black,fill=#1] (0,0) circle (0.1);
    \foreach \x in {-2,-1,0,1,2}{
      \fill[draw=black,fill=white] (\x/4,2) circle (0.1);
      }
    }}

\begin{figure}[ht!]
  \begin{center}
  \begin{tikzpicture}[xscale=0.5,yscale=0.2]
      \fill[black!20!white] (0,3) circle (11.5);
    \fill[green,opacity=0.30] (-3,-5.8) -- (3,5.8) -- (18,-2.2) -- (13,-14) -- cycle; 
      \fill[fill=black!15!white] (0,0) circle (6.3);
    \foreach \x/\y in {0/2,2/-0.2,-2/-0.2}{
      \draw[line width=1,black] (0,0) -- (\x,\y);
    }
    \foreach \xa/\ya/\xb/\yb in {-4.8/-3/-8/-2,-4.8/-3/-10/2,-4.8/-1/-9/4,-4.8/-1/-8/1,-4.8/1/-7/3,-4.8/1/-8/6,-4.8/3/-6/5,-4.8/3/-6/10,-3.7/3.8/-4/8,-3.7/3.8/-2/12,-1.3/3.8/-2/8,-1.3/3.8/-1/10}{
      \draw[line width=0.5,rounded corners] (\xa,\ya) -- (\xb,\yb);
      \draw[line width=0.5,rounded corners] (-\xa-0.2,\ya) -- (-\xb,\yb);
    }
    \draw[line width=0.5,rounded corners] (-8,-2) -- (-9,-1) --  (-11,2) -- (-11,4) -- (-8,11) -- (-4,13) -- (-3,13) -- (-2,12);
    \draw[line width=0.5,rounded corners] (-2,12) -- (-1,13) -- (0,12);
    \draw[line width=0.5,rounded corners] (0,12) -- (1,13) -- (3.5,13.5) -- (3,9) -- (2,8);
    \draw[line width=0.5,rounded corners] (0,12) -- (1,11);
    \draw[line width=0.5,rounded corners] (8,6) -- (7,7) -- (7,11.5) -- (5,12) -- (5,9) -- (4,8);
    \foreach \x/\y in {-8/-2,-10/2,-9/4,-8/1,-7/3,-8/6,-6/5,-6/10,-4/8,-2/12,-2/8,-1/10,0/12}{
      \draw[line width=1, draw=red] (\x,\y) -- +(-100:1);
      \draw[line width=1, draw=blue] (\x,\y) -- +(-70:1);
      \draw[line width=1, draw=blue] (\x,\y) -- +(-80:1);
      \draw[line width=0.5,rounded corners] (\x,\y) -- +(135:1);
      \draw[line width=0.5,rounded corners] (\x,\y) -- +(45:1);
      \fill[fill=white,draw=black] (\x,\y) circle (0.2);
      \draw[line width=1, draw=red] (-\x,\y) -- +(-100:1);
      \draw[line width=1, draw=blue] (-\x,\y) -- +(-70:1);
      \draw[line width=1, draw=blue] (-\x,\y) -- +(-80:1);
      \draw[line width=0.5,rounded corners] (-\x,\y) -- +(135:1);
      \draw[line width=0.5,rounded corners] (-\x,\y) -- +(45:1);
      \fill[fill=white,draw=black] (-\x,\y) circle (0.2);
    }
    \draw node at (0.7,-1.8) {\bigpicturebisdecoratednode{0.2}{0}};
    \draw node at (-1.2,-9) {\bigpicturebisdecoratednode{0.2}{70}};
    \draw node at (1.8,-11.5) {\bigpicturebisdecoratednode{0.2}{70}};
    \draw node at (-1.2,-13) {\bigpicturebisdecoratednode{0.2}{40}};
    \draw node at (-4.3,-13.4) {\bigpicturebisdecoratednode{0.2}{30}};
    \draw node[anchor=south] at (0.1,0) {\tikz[xscale=0.6,yscale=0.2]{\twotreebis{black}{-100}}};
    \draw node[anchor=west] at (1.5,-0.6) {\tikz[xscale=0.6,yscale=0.2,rotate=-90]{\twotreebis{black}{-10}}};  
    \draw node[anchor=east] at (-1.5,-0.6) {\tikz[xscale=0.6,yscale=0.2,rotate=90]{\twotreebis{black}{170}}};
    \draw[line width=1,red,decorate,decoration=snake] (0,0) -- (-2,-9.6);
    \fill[fill=white,draw=black] (0,0) circle (0.3);
    \draw node at (-0.5,0.75) {$u$};
    \draw[line width=2,opacity=0.5,green] (0,0) circle (6.3);
    \draw node[anchor=north] at (-2,-9) {\tikz[scale=0.5,rotate=180]{\threetreebis{red}}};
    \draw node at (0,15.9) {\emph{View with decorations}};
    \draw node at (15,-8) {\tikz[xscale=0.5,yscale=0.2]{
        \fill[fill=black!15!white] (0,0) circle (6.3);
        \draw[line width=2,opacity=0.5,green] (0,0) circle (6.3);
        }};
    \draw[line width=1] (15,-10) -- (14,-10);
    \draw[line width=1] (15,-10) -- (16,-10);
    \draw[line width=1] (15,-10) -- (15,-8);
    \draw node at (15,-6) {\mytree{white}};
    \draw node[rotate=90] at (13,-10) {\mytree{white}};
    \draw node[rotate=270] at (17,-10) {\mytree{white}};
    \draw node at (15,-10) {\tikz{\fill[white,draw=black] (15,-10) circle (0.05);}};
    \draw node at (15,-10.9) {$u$};
    \draw node at (19.5,-16) {\begin{minipage}{7cm}\emph{Simplified view} \end{minipage}};
  \end{tikzpicture}
  \caption{Pictorial representation of the graph $G$, with the expander $E$ shown inside the gray ellipse (circular vertices) and details of the self-similar trees attached to a vertex $u$ (square vertices). The expander is locally tree-like within a ball of radius less than half the girth around~$u$ (shown by the green boundary). The graph~$G$ locally resembles a self-similar tree -- see \Cref{Lem:locss}. The simplified view depicts the local tree structure around $u$, drawing the incident edges within the expander, but omitting the decorations arising from the blue and red edges.}
  \label{fig:overview}
  \end{center}
\end{figure}

\paragraph{Hardness of classical exploration.}
The graph~$G$ is designed so that it is difficult for classical algorithms to distinguish locally between the expander part and the tree part (\Cref{Lem:locss,Lem:new-hardmaing}). Indeed, since the expander has high girth and the final graph is degree-regular except at the leaves, it is difficult to distinguish the expander vertices from the tree vertices in a small neighborhood. Moreover, the self-similar trees have a sufficiently confusing structure that makes it difficult to infer one's position in the graph from their leaves (\Cref{Lem:hardss} and \Cref{App:trees}). For instance, a simple random walk on~$G$ is insufficient for sampling from the ground state of~$H$, since its stationary distribution is the degree distribution rather than the ground-state distribution, with the latter being far more concentrated on the vertices close to the expander. While more sophisticated classical algorithms might sample from the correct distribution, our main result is to show that none can do so efficiently.
 
\paragraph{Localization property.}
The classical hardness is established by examining a specific \emph{localization} property of the ground state~$\psh$, which must also be approximately satisfied by the guiding state~$\psin$ and the output state~$\psout$. This property states that, upon sampling sub-exponentially many vertices from the ground-state distribution, the distance between their \emph{nearest vertices} on the expander graph~$E$ is large with very high probability (\Cref{Lem:hfaridealpair}). This follows essentially from the fact that the marginal ground-state distribution over the expander vertices is uniform (\Cref{Lem:spec}). The localization property must also hold for the distributions arising from~$\psin$ and~$\psout$ with constant probability (\Cref{Prop:hfarideal}).

\paragraph{Hardness of classical sampling.}
We demonstrate that any classical algorithm making fewer than~$2^{m^c}$ queries to~$H$ for some~$c < 1$ (i.e., sub-exponentially many queries) cannot reproduce the above localization property in its output distribution with more than exponentially small probability (\Cref{Prop:classhard}). By standard arguments, a classical algorithm can be assumed to explore the graph~$G$ by growing connected components around the input sampled vertices (\Cref{lem:conncomp}). However, since the graph is difficult to explore, none of these connected components can reach an expander vertex sufficiently far from its starting point (\Cref{Lem:new-hardmaing,Lem:hardss}). It follows that such an algorithm cannot satisfy the localization property with high probability, and therefore cannot solve the GGSP problem on~$H$.

\subsubsection{Guided Ground-State Energy Estimation}
We provide a quick overview of the main ideas behind the proof of \Cref{Thm:GEEHard}, highlighting the new ingredients needed there. The lower bound is proved by constructing two graph families whose negative adjacency matrices define two families of stoquastic Hamiltonians corresponding to the $\yes$ and $\no$ instances of GGSE. Thus, solving GGSE on these instances amounts to distinguishing between the two graphs.

\paragraph{Main and envelope graphs.}
The graph chosen for the YES instance is the same \emph{main graph}~$G$ used in the GGSP lower bound. The graph for the NO instance, which we call the \emph{envelope graph}~$F$ (\Cref{Def:maingraph}), is designed to be locally similar to~$G$ while being spectrally separated from it. Like~$G$, it is obtained from a high-degree, high-girth graph~$E'$ by attaching many copies of self-similar trees~$\treh$ to its vertices. Structurally, the main difference between~$G$ and~$F$ is that the base graph~$E'$ of $F$ has smaller degree than the base graph~$E$ of $G$ (which gives it a smaller top eigenvalue), while additional self-similar trees are attached to it. In particular, the internal vertices of~$F$ have the same degree as those of~$G$. The additional trees are chosen so that the local structure of~$F$ mimics that of~$G$ up to a certain distance, making the two graphs locally indistinguishable.

\paragraph{Local indistinguishability.}
The graphs are constructed so that a classical algorithm \emph{starting at an expander vertex} cannot distinguish $G$ from $F$ without exploring sufficiently far. The key novelty establishing this property is an embedding of the local graph structure around any expander vertex of $G$ into the corresponding structure in $F$ (\Cref{prop:embed-local}). The two local structures agree up to a large distance. Consequently, the adjacency-list responses seen by a local exploration have the same distribution in the two graphs until the exploration reaches a sufficiently far-away vertex. The hardness of such an exploration then follows by the same arguments as in the GGSP lower bound. This holds even when multiple expander vertices are provided as input, as long as they are sufficiently far apart from one another (\Cref{Prop:distgraph}).

\paragraph{Hardness of classical energy estimation.}
Finally, it remains to show that the graphs are indistinguishable even to algorithms given a guiding state distribution. Unlike in the GGSP lower bound, this distribution cannot be chosen arbitrarily, since it may place too much probability on vertices whose local structures are distinguishable. We show that the uniform superposition over the expander vertices is a valid guiding state, with overlap $1-O(1/n)$ with the ground state (\Cref{Lem:spec}). Under this distribution, the input samples are far-away expander vertices with high probability, which lets us apply the hardness result above. We conclude that the GGSE problem on $G$ and $F$ cannot be solved with probability exceeding $1/2$ by more than a sub-exponentially small amount.

\section{Setup and graph construction}

  \subsection{Guided Ground-State Preparation (GGSP) problem}
  The Guided Ground-State Preparation problem (GGSP) is formally defined below. In the quantum setting, the goal is to prepare the ground state of a gapped Hamiltonian, given access to copies of a guiding state. In the classical setting, quantum states are replaced by classical samples drawn from the corresponding standard basis measurement distribution. We use $\poly(m)$ as shorthand for $m^c$ for some absolute constant $c$. The constants $3/4$ and $1/5$ in the output condition are chosen arbitrarily.

\begin{definition}
  \label{Def:GGSP}
  The \emph{Guided Ground-State Preparation} (GGSP) problem is defined as follows. We are given an $m$-qubit stoquastic Hamiltonian $H$ with bounded norm $\norm{H} \leq \poly(m)$, which has a unique \emph{ground state} $\psh$ and a spectral gap at least $1/\poly(m)$. In addition, we are given the description of $t$ copies of a \emph{guiding state} $\psin$ that satisfies $\abs{\ip{\psi_{\mathrm{in}}}{\psi_H}} \geq 1/\poly(m)$. We should output the description of a \emph{refined state} $\psout$ such that $\abs{\ip{\psi_{\mathrm{out}}}{\psi_H}} \geq 3/4$. The descriptions of the guiding and refined states are given in one of the following two models:
  \begin{itemize}
    \item {\it (Quantum input-output state)} The algorithm is given the quantum state $\psin^{\otimes t}$ and must output the quantum state $\psout$.
    \item {\it (Classical input-output sample)} Let $\pin$ and $\pout$ denote the probability distributions on $\rn^m$ obtained by measuring $\psin$ and $\psout$, respectively, in the standard basis. The algorithm is given i.i.d. samples $x_1,\dots,x_{t} \sim \pin$, and must output one sample $y \sim \pout$ such that the joint distribution of $(x_1,\dots,x_{t},y)$ is $1/5$-close in total variation distance to the product distribution $\pin^{\otimes t} \times \pout$.
  \end{itemize}
\end{definition}

In the classical setting, the conditions on the input-output states could equivalently be defined in terms of the \emph{fidelities} $F(\pin,\ph) = \abs{\ip{\psi_{\mathrm{in}}}{\psi_H}}^2 \geq 1/\poly(m)$ and $F(\pout,\ph) = \abs{\ip{\psi_{\mathrm{out}}}{\psi_H}}^2 \geq (3/4)^2$, where $\ph(x) = \abs{\ip{x}{\psi_H}}^2$ is the ground-state distribution.

The GGSP problem can be solved quantumly by filtering out the ground-state component of the guiding state. Because this state already has non-negligible overlap with the ground state, quantum algorithms can rapidly amplify this overlap, something classical algorithms cannot do efficiently in general. Below, we describe an elementary approach to achieve this goal.

\begin{proposition}[\sc Quantum algorithm for the GGSP problem]
  \label{prop:qalgo}
  There exists a quantum algorithm that solves the GGSP problem by outputting the exact ground state $\psout = \psh$ of $H$ with cost $\poly(m) \cdot T_H$, where $T_H$ is the cost of simulating the unit-time evolution $e^{-iH}$.
\end{proposition}

\begin{proof}
  The algorithm runs quantum phase estimation on $H$ using the guiding state $\psin$ as input. Because the spectral gap is at least inverse polynomial, it suffices to use inverse-polynomial precision to distinguish the smallest-eigenvalue component $\psh$ from the others. Measuring the energy register then collapses the state to that component with probability equal to the squared overlap $\abs{\ip{\psi_{\mathrm{in}}}{\psi_H}}^2 \geq 1/\poly(m)$. The success probability can be amplified arbitrarily close to 1 by repeating the procedure on fresh copies of $\psin$ until the state collapses to $\psh$.
\end{proof}

The simulation cost $T_H$ appearing in the quantum algorithm above depends on the properties and input specification of $H$. It is polynomial in $m$ for several standard classes of Hamiltonians, including local Hamiltonians, sparse Hamiltonians given oracle access to their nonzero entries (as in the present work), and, more generally, Hamiltonians admitting an efficient block-encoding.

\subsubsection*{On the independence requirement.}

A subtle but important requirement of \Cref{Def:GGSP} -- in the classical setting -- is that the output sample $y$ must be (almost) independent of the input samples $x_1,\dots,x_t$. Without this condition, an algorithm could sometimes reuse the same input as output (for example, when the guiding state already coincides with the ground state). This would pose a significant practical limitation. For instance, standard concentration bounds, such as Chernoff or Chebyshev inequalities, rely on (nearly) independent samples to hold. Producing multiple correlated samples would therefore prevent estimation of statistics of the ground-state distribution via such bounds: if we estimate an expectation value $\bra{\psi_H} O \ket{\psi_H}$ via the empirical mean $\frac{1}{t} \sum_i O(y_i)$, where $y_1,\dots,y_t$ are correlated samples, then the estimator's variance may not shrink with $t$ at all. This obstruction is formalized by our subsequent lower bound for ground-state energy estimation (\Cref{Thm:GEEHard}), which shows that estimating $\bra{\psi_H} H \ket{\psi_H}$ is no easier than outputting independent samples.

A similar consideration is standard in the classical Markov chain Monte Carlo (MCMC) literature. For instance, when estimating partition functions via simulated annealing, Markov chains are initialized from a warm-start distribution (playing the role of the guiding state) and run long enough to produce fresh, independent samples from the stationary distribution (see, e.g., \cite[Section 6]{JSV04j} and \cite[Section 7]{SVV09j}). 

One may wonder why we do not impose a similar independence requirement in the quantum setting. Due to the no-cloning theorem, it is generally impossible to produce the ground state without entirely consuming a copy of the guiding state. Nevertheless, this obstruction does not preclude certain types of quantum statistical estimators that behave differently from classical ones. For example, quantum speedups for MCMC methods can rely on non-destructive quantum primitives~\cite{HW20c,CH23c}, which make it possible to use a single copy of the ground state~$\psh$ restored at the end of the computation, together with multiple applications of the reflection~$\id - 2 \proj{\psi_H}$.

  \subsection{Guided Ground-State Energy Estimation (GGSE) problem}
  The Guided Ground-State Energy-Estimation problem (GGSE) is formally defined below. 

\begin{definition}[\sc Guided Ground-State Energy Estimation]
    \label{Def:ggee}
    The \emph{Guided Ground-State Energy Estimation} (GGSE) problem is defined as follows. We are given oracle access to an $m$-qubit stoquastic Hamiltonian $H$ satisfying $\norm{H}\leq\poly(m)$, with a ground state $\psh$. In addition, we are given the description of $t$ copies of a \emph{guiding state} $\psin$ that satisfies $\abs{\ip{\psi_{\mathrm{in}}}{\psi_H}} \geq 1/\poly(m)$ (in the same input model as in \Cref{Def:GGSP}). The lowest eigenvalue $\lambda_{\rm min}(H)$ of the Hamiltonian is promised to satisfy exactly one of the following:
      \[
      \textnormal{YES: } \lambda_{\rm min}(H) \leq a \qquad\textnormal{NO: } \lambda_{\rm min}(H) \geq b,
    \]
    where $b-a\geq1/\poly(m)$. The goal is to output $1$ in the YES case and $0$ in the NO case on every input satisfying the promise.
\end{definition}

For the sake of completeness, we provide a quick proof of the fact that there exists an efficient quantum algorithm for this problem by applying the quantum phase estimation algorithm.

\begin{proposition}[\sc Quantum algorithm for GGSE]\label{prop:q-ggse}
  There exists a quantum algorithm that solves the GGSE problem with success probability at least $2/3$ and at cost $\poly(m) \cdot T_H$, where~$T_H$ is the cost of simulating the unit-time evolution $e^{-iH}$.
\end{proposition}

\begin{proof}
    The algorithm applies quantum phase estimation to $H$ with input $\psin$ and precision below $b-a$. Repeating the procedure with polynomially many input copies and taking the minimum measured energy yields the correct ground-state energy with constant probability. We note that, unlike in GGSP, the algorithm does not require $H$ to have a large spectral gap.
\end{proof}

  \subsection{Graph construction}
  We describe the graph $G_{n,E}$ that underlies both of our lower bounds. The graph combines a high-degree, high-girth spectral expander with a particular family of self-similar trees. We use a free integer parameter $n$ throughout the construction, which will later be chosen as $n=\poly(m)$ when defining the corresponding Hamiltonians.

We first state the properties required of the spectral expanders. A key difficulty is that the vertex degree must grow (logarithmically) with the number of vertices, which rules out constructions with constant degree.

\begin{proposition}[\sc High-degree, high-girth expanders $E_n$]
  \label{Prop:expander}
  There exists a family of regular graphs $(E_n)_n$ such that, for even $n$ large enough, $E_n$ has $N_E \in [2^{100n^{5/2}\log^2 n},2^{n^3 \log^2 n}]$ vertices, degree~$d_n = n$, spectral gap at least $n/2$ and girth at least~$g_n = 100n^{5/2}\log{n}$.
\end{proposition}

\begin{proof}
  We provide two constructions of such graphs in \Cref{App:app-graph}, one based on the number-theoretic Ramanujan graphs from \cite{LPS88j}, and the other on the short cycle removal algorithm from \cite{Par21c}. The first construction offers a more explicit description of the expander, but only works for even degrees. The second construction provides an algorithm that runs in time polynomial in the number of vertices.
\end{proof}

We expect that the graph constructed in~\cite{GHV21c} can be derandomized by relying on similar expanders as above, together with high-girth bipartite regular graphs (such as those introduced in \cite{Dah14j}) in place of the random bipartite matchings used in that work.

If the graphs do not need to be expanders, as is the case in \Cref{sec:energyhard} for GGSE, other constructions of high-girth graphs become available. We describe such a sequence of graphs that is used later.

\begin{proposition}[\sc High-degree, high-girth graphs $E'_n$]
  \label{Prop:highgirth}
  There exists a family of regular graphs $(E'_n)_n$ such that, for $n$ large enough, $E'_n$ has $N_E \in [2^{100n^{5/2}\log^2 n},2^{n^3\log^2 n}]$ vertices, degree~$n-\sqrt{n}$ and girth at least $g_n = 100n^{5/2}\log n$.
\end{proposition}

\begin{proof}
  An early result of Erd\H{o}s and Sachs~\cite{ES63j,EJ08j} establishes the existence of $d$-regular graphs on $N$ vertices with girth at least $g$, provided that $d \geq 3$, $g \geq 4$, $N \geq 4(d-1)^{g-1}$, and~$N$ is even. A random $d$-regular graph has nonzero probability, albeit very small, of having girth at least $g$ whenever $(d-1)^{2g-1} = o(N)$ \cite[Corollary 1]{MWW04j}. A combinatorial construction of high-girth Cayley graphs is also given in~\cite{Dah14j}.
\end{proof}

Next, we describe a family of trees with a fractal-like structure, which we call \emph{self-similar trees}. The construction is based on the notion of \emph{graph decoration}: given two graphs~$G$ and~$G'$, one attaches a distinct copy of~$G'$ to each vertex of the base graph~$G$ by adding a new edge connecting that vertex to a distinguished vertex of the corresponding copy of~$G'$. Such a process has been studied in the context of stoquastic Hamiltonian constructions in \cite{Has21j,GHV21c}, and our construction of the trees bears strong similarities to these works. One difference is that we do not decorate the leaves of the base graph. This has the advantage of making the final graph slightly smaller, and ensures that it is degree-regular everywhere except at its root and its leaves.

\begin{definition}[\sc Self-similar trees $\treh_{n,k}$]
  \label{Def:selfsim}
  For $n \geq 1$ and $1 \leq k \leq \sqrt{n}+1$, consider the degree sequence $d_{n,k} = 2n - k\sqrt{n}$ and the depth sequence $\ell_{n,k} = k\cdot10n^{3/2}\log{n}$. Define the following trees,
    \begin{itemize}
      \item The \emph{perfect tree} $C_{n,k}$ is the tree in which every internal vertex has exactly $d_{n,k}-1$ children and all leaves are at depth $\ell_{n,k}$. The \emph{$0$-level decorated tree} $T_{n,k,0}$ is defined to be $C_{n,k}$.
      \item The \emph{$r$-level decorated tree} $T_{n,k,r}$ is defined iteratively, for $r = 1,\dots, k-1$, by adding $d_{n,k - r} - d_{n,k-r+1}$ edges to each internal node of $T_{n,k,r-1}$ and attaching the other endpoint of each edge to the root of a new copy of $C_{n,k-r}$.
      \item The \emph{full-level decorated tree} $\treh_{n,k}$ is defined to be $T_{n,k,k-1}$.
    \end{itemize}
  The \emph{$r$-level leaves} is the set of all new leaves added in the construction of $T_{n,k,r}$ by decorating~$T_{n,k,r-1}$.
\end{definition}

The above definition is equivalent to the following construction, which will prove convenient for some of the later analysis.

\begin{fact}[\sc Bottom-up construction of $\treh_{n,k}$]
  \label{Fac:botup}
  The full-level decorated tree $\treh_{n,k}$ can equivalently be obtained by, for each $r = 1,\dots, k-1$, adding $d_{n,r} - d_{n,r+1} = \sqrt{n}$ edges to each internal node of $C_{n,k}$ and attaching the other endpoint of each edge to the root of a new copy of $\treh_{n,r}$. The base case is $\treh_{n,1} = C_{n,1}$.
\end{fact}

These trees have the following properties.

\begin{fact}
  All vertices in $T_{n,k,r}$ (except the root and the leaves) have degree $d_{n,k-r}$. The number of vertices in $\treh_{n,k}$ is at most $2^{12n^{3}\log^2{n}}$.
\end{fact}

We now describe the graph $G_{n,E}$, which is obtained by decorating the expander defined above with self-similar trees. We also add dummy isolated vertices for technical reasons, following prior work on classical graph exploration (e.g.,~\cite{CCD03c}). These isolated vertices are used to force classical algorithms to explore a connected component, since otherwise a queried vertex would be isolated with very high probability.

\begin{definition}[\sc Decorated graph $G_{n,E}$]
  \label{Def:maingraph}
  Given an integer $n \geq 1$, define the graph constructed as follows:
  \begin{itemize}
    \item[] Fix a pair $(K,E) \in \{(\sqrt n,E_n),(\sqrt n+1,E'_n)\}$ (see \Cref{Prop:expander,Prop:highgirth}).
    \item[] For $k = 1,\dots, K-1$, add $d_{n,k} - d_{n,k+1} = \sqrt{n}$ edges to each node of $E$ and attach the other endpoint of each edge to the root of a new copy of $\treh_{n,k}$ (see \Cref{Def:selfsim}).
    \item[] Pad the construction with $2^{15n^{3}\log^2{n}} - N_E \cdot (1 + \sqrt{n}\cdot\sum_{k=1}^{K-1}|\treh_{n,k}|)$ isolated vertices.
  \end{itemize}
  We let $G_{n,E}$ denote the graph constructed above. We call $G_n = G_{n,E_n}$ and~$F_n = G_{n,E'_n}$ the \emph{main graph} and the \emph{envelope graph}, respectively.

  We represent the vertices of the graph $G_{n,E}$ by the sets $V = V_E \cup V_T \cup V_I$, where $V_E$ are the vertices coming from the subgraph $E$ (called the \emph{expander vertices}), $V_T$ are the vertices coming from the trees~$\treh_{n,k}$ (called the \emph{tree vertices}), and $V_I$ are the other \emph{isolated vertices}. When there is risk of ambiguity between different graphs, we write $V(G_{n,E}) = V$ (and similarly for the other vertex sets).

  For each non-isolated vertex $v \in V_E \cup V_T$, we let $\ep{v} \in V_E$ denote the closest expander vertex to $v$ in the graph (this is $v$ itself when $v \in V_E$, and the expander vertex to which the tree containing~$v$ is attached when~$v \in V_T$).

  The adjacency matrix of $G_{n,E}$ is denoted by $A_{n,E}$, and its top eigenvector by $\ket{\psi_{n,E}}$.
\end{definition}

\begin{fact}
  \label{Fac:grp}
  All vertices in $G_{n,E}$ (except the leaves and the isolated vertices) have degree $d_{n,1}$. The graph $G_{n,E}$ has the same girth as $E$. The number of vertices is $\abs{V} = 2^{15n^{3}\log^2{n}}$.
\end{fact}

There are two crucial properties of $G_{n,E}$ that will play a central role in both of our lower-bound arguments. First, $G_{n,E}$ inherits the spectral properties of the subgraph $E$, in the sense that it has a large spectral gap if $E$ does, and that its top eigenvector has a large overlap with the top eigenvector of $E$. We use the notation $\psi$ (rather than $\ket{\psi}$) to denote a state that is not necessarily normalized.

\begin{lemma}[\sc Spectral properties of the adjacency matrix $A_{n,E}$]
  \label{Lem:spec}
  The difference between the two largest eigenvalues of $A_{n,E_n}$ is at least $d_n/2 - 4\sqrt{2n}$.

  The top eigenvector $\psi_{n,E}$ of $A_{n,E}$, where $E \in \set{E_n,E'_n}$, can be written as $\psi_{n,E} = \psi_E + \sum_{v\in V_E}\psi_{v,T}$, where $\psi_E$ is a uniform (unnormalized) superposition over the expander vertices $\set{{v}}_{v \in V_E}$, and each $\psi_{v,T}$ is supported on the vertices of the trees attached to $v$.
  Moreover, the states $\{\psi_{v,T}\}_{v\in V_E}$ are identical up to a permutation of the computational basis.
  
  The base component satisfies $\norm{\psi_E}^2 \geq (1-O(1/n))\norm{\psi_{n,E}}^2$.
\end{lemma}

\begin{proof}
  The proof is deferred to \Cref{App:app-spectral}. It relies on generic spectral properties of decorated graphs, similar to those used in~\cite{Has21j,GHV21c}.
\end{proof}

The second property is that, locally, $G_{n,E}$ resembles the tree $\treh_{n,K}$ along any edge within the expander graph (for the remainder of this section, $K$ denotes the index paired with $E$ in \Cref{Def:maingraph}). This will make it very difficult for classical algorithms to distinguish between these two graphs, although their spectral properties are radically different. We illustrate a simplified version of that in \Cref{fig:overview,Fig:expanderlocal} and in the next lemma. A stronger version is given in \Cref{fig:zoom} and \Cref{Lem:new-hardmaing}, where we consider multiple connected components rooted at arbitrary vertices. The property is based on the following distance measure $\dist_E$ defined on the graph~$G_{n,E}$.

\begin{definition}[\sc Distance measure]
  \label{Def:loc}
  Given two vertices $u,v \in V$, the distance~$\dist_E(u,v)$ is defined as the number of \emph{expander} edges along a shortest path connecting $u$ and $v$ in~$G_{n,E}$ (with the convention $\dist_E(u,v) = +\infty$ if either $u$ or $v$ is an isolated vertex). Given a subset~$U \subseteq V$ of vertices, we define $\dist_E(U,v) = \min_{u \in U} \dist_E(u,v)$. 
\end{definition}

Note that $\dist_E$ is not a metric, since the distance between two distinct vertices $u \neq v$ can be zero (if they belong to trees attached to the same expander vertex). Nevertheless, it satisfies the triangle inequality. We use $\dist_E$ to study locality in the graph $G_{n,E}$ as follows.

In most of our arguments we work with a \emph{truncated neighborhood} of a set of expander vertices~$\mcl{U}$. This is the subgraph of $G_{n,E}$ induced by all vertices (expander or tree) within distance~$r$ of~$\mcl{U}$, together with only the expander vertices at distance exactly $r+1$. We define it formally here.

\begin{definition}[\sc Truncated neighborhood]\label{def:trunc-neigh}
  Let $G \in \set{G_n, F_n}$ be the main or the envelope graph. For a set of expander vertices $\mcl{U} \subseteq V_E(G)$ and an integer $r \ge 0$, let $B_{G,r}(\mcl{U})$ be the subgraph of $G$ induced by
  $$\set{v \in V_E(G) \cup V_T(G): \dist_E(\mcl{U},v)\leq r} \cup \set{v \in V_E(G): \dist_E(\mcl{U},v)=r+1}.$$
  We write $B_{G,r}(v)$ for $B_{G,r}(\set{v})$ and drop $G$ when it is clear from context.
\end{definition}

\begin{lemma}[\sc Local structure of $G_{n,E}$]
  \label{Lem:locss}
  Consider the graph $G_{n,E}$, where $K$ is the index paired with $E$ in \Cref{Def:maingraph}, and let $u \in V_E$ be an expander vertex. Then the truncated neighborhood $B_{\ell_{n,K}}(u)$ (\Cref{def:trunc-neigh}) is isomorphic to the following graph (\Cref{Fig:expanderlocal}): a central vertex~$\td{u}$ joined by an edge to the root of each of $\sqrt{n}$ copies of $\treh_{n,k}$ for every $k = 1,\dots,K-1$, and of $d_{n,K}$ copies of $\treh_{n,K}$.

  Moreover, the isomorphism can be chosen so that $u$ maps to $\td{u}$ and the expander vertices at distance exactly $\ell_{n,K}+1$ from $u$ map to the $0$-level leaves of the copies of $\treh_{n,K}$.
\end{lemma}

\begin{proof}
  The subgraph $B_{\ell_{n,K}}(u)$ is necessarily a tree, since the subgraph $E$ has no cycles of length less than~$g_n$ by the girth property, and $g_n > 2(\ell_{n,K} + 1) + 1$. Hence, the expander vertices in~$B_{\ell_{n,K}}(u)$ form a full tree of depth $\ell_{n,K} + 1$ and degree~$d_{n,K}$, to which are attached the decorations specified in \Cref{Def:maingraph}. Observe that, for any expander neighbor $v$ of $u$, the subtree rooted at $v$ matches exactly the definition of the self-similar tree $\treh_{n,K}$ from \Cref{Def:selfsim} and \Cref{Fac:botup}.
\end{proof}

\begin{figure}[t]
  \begin{center}

\begin{tikzpicture}[
every node/.style={font=\small},
centre/.style={
circle,
draw=violet!75!black,
thick,
fill=magenta!35!white,
text=violet!75!black,
minimum size=7mm,
inner sep=0pt
},
common/.style={red!75!black,thick},
expander/.style={green!45!black,thick},
extra/.style={blue!75!black,thick}
]

%
\newcommand{\outtree}[6]{%
\pgfmathsetmacro{\apexdistance}{1.05}
\pgfmathsetmacro{\basedistance}{1.05+#5}
\pgfmathsetmacro{\labeldistance}{1.05+0.58*#5}

\coordinate (treeapex)
at ($(#1.center)+(#2:\apexdistance)$);

\coordinate (treebase)
at ($(#1.center)+(#2:\basedistance)$);

\coordinate (treebaseup)
at ($(treebase)+({#2+90}:#6)$);

\coordinate (treebasedown)
at ($(treebase)+({#2-90}:#6)$);

\draw[#3]
(#1.#2) -- (treeapex);

\draw[#3]
(treeapex)
-- (treebaseup)
-- (treebasedown)
-- cycle;

\pgfmathtruncatemacro{\labelflip}
{ifthenelse(cos(#2)<0,180,0)}

\node[
rotate={#2+\labelflip},
#3,
font=\scriptsize
] at ($(#1.center)+(#2:\labeldistance)$) {#4};
}

%
\newcommand{\outtreecut}[7]{%
\pgfmathsetmacro{\apexdistance}{1.05}
\pgfmathsetmacro{\basedistance}{1.05+#5}
\pgfmathsetmacro{\labeldistance}{1.05+0.58*#5}

\coordinate (treeapex)
at ($(#1.center)+(#2:\apexdistance)$);

\coordinate (treebase)
at ($(#1.center)+(#2:\basedistance)$);

\coordinate (treebaseup)
at ($(treebase)+({#2+90}:#6)$);

\coordinate (treebasedown)
at ($(treebase)+({#2-90}:#6)$);

\draw[#3,dashed]
(#1.#2) -- (treeapex);

\draw[#3,dashed]
(treeapex)
-- (treebaseup)
-- (treebasedown)
-- cycle;

\begin{scope}
\clip (#1.center) circle[radius=#7];

\draw[#3]
(#1.#2) -- (treeapex);

\draw[#3]
(treeapex)
-- (treebaseup)
-- (treebasedown)
-- cycle;
\end{scope}

\pgfmathtruncatemacro{\labelflip}
{ifthenelse(cos(#2)<0,180,0)}

\node[
rotate={#2+\labelflip},
#3,
font=\scriptsize
] at ($(#1.center)+(#2:\labeldistance)$) {#4};
}

%
\newcommand{\branchcount}[6]{%
\begin{scope}[shift={(#1.center)}]

\draw[#6]
(#2:#4)
arc[
start angle=#2,
end angle=#3,
radius=#4
];

\pgfmathsetmacro{\middleangle}{(#2+#3)/2}

\node[
#6,
inner sep=0pt
] at (\middleangle:{#4+0.35}) {#5};

\end{scope}
}


\node[font=\large] at (0,4.5) {$B_{\ell_{n,K}}(u)$};
\node[centre] (u) at (0,0) {$u$};


\outtree{u}{120}{common}
{$\treh_{n,1}$}
{1.20}{0.7}

\node[common] at (-1.4,1.2) {$\iddots$};

\outtree{u}{170}{common}
{$\treh_{n,1}$}
{1.20}{0.7}

\branchcount{u}{170}{120}{0.72}
{$\sqrt{n}$}{common}

\node[common] at (0,2.15) {$\dots$};


\outtree{u}{60}{common}
{$\treh_{n,K-1}$}
{2}{0.60}

\node[common] at (1.45,1.35) {$\ddots$};

\outtree{u}{20}{common}
{$\treh_{n,K-1}$}
{2}{0.60}

\branchcount{u}{60}{20}{0.72}
{$\sqrt{n}$}{common}


\outtree{u}{-35}{extra}
{$\treh_{n,K}$}
{2.70}{0.57}

\node[extra] at (0,-2.15) {$\dots$};

\outtree{u}{215}{extra}
{$\treh_{n,K}$}
{2.70}{0.57}



\coordinate (KRbase) at ($(u.center)+(-35:3.75)$);
\coordinate (KRup)   at ($(KRbase)+(55:0.57)$);
\coordinate (KRmid)  at (KRbase);
\coordinate (KRdown) at ($(KRbase)+(-125:0.57)$);

\foreach \P in {KRup,KRmid,KRdown}{
    \foreach \ang in {-50,-30,-10}{
        \draw[extra,dashed] (\P) -- ++(\ang:0.75);
    }
}


\coordinate (KLbase) at ($(u.center)+(215:3.75)$);
\coordinate (KLup)   at ($(KLbase)+(305:0.57)$);
\coordinate (KLmid)  at (KLbase);
\coordinate (KLdown) at ($(KLbase)+(125:0.57)$);

\foreach \P in {KLup,KLmid,KLdown}{
    \foreach \ang in {190,210,230}{
        \draw[extra,dashed] (\P) -- ++(\ang:0.75);
    }
}

\branchcount{u}{-145}{-35}{0.72}
{$d_{n,K}$}{extra}


\begin{pgfonlayer}{background}


\fill[red!4]
(u.center) circle[radius=4cm];

\fill[green!3]
(u.center)
-- ++(200:4cm)
arc[
start angle=200,
end angle=340,
radius=4cm
]
-- cycle;

\draw[gray!45,thin]
(u.center) circle[radius=4cm];

\draw[gray!35,thin]
($(u.center)+(200:4cm)$)
-- (u.center)
-- ($(u.center)+(340:4cm)$);

\end{pgfonlayer}


\node at (0.25,-4.5) {
\begin{tikzpicture}[baseline=-0.6ex]

\fill[green!3]
(0,0) rectangle (0.35,0.22);
\draw[gray!75,thin]
(0,0) rectangle (0.35,0.22);
\node[anchor=west]
at (0.45,0.11) {Expander part};

\fill[red!4]
(4.0,0) rectangle (4.35,0.22);
\draw[gray!75,thin]
(4.0,0) rectangle (4.35,0.22);
\node[anchor=west]
at (4.45,0.11) {Decorations};

\end{tikzpicture}
};

\end{tikzpicture}
    \caption{The truncated neighborhood $B_{\ell_{n,K}}(u)$ of an expander vertex $u \in V_{E}$.}
  \label{Fig:expanderlocal}
  \end{center}
\end{figure}
\color{black}

  \subsection{Hamiltonian construction}
  We describe the stoquastic Hamiltonian encoding used throughout the paper. Given a graph~$G$ on $2^m$ vertices, we associate to it the negative of its adjacency matrix, with the vertices encoded via a permutation. In our classical lower bounds, this permutation will be chosen uniformly at random. The random labeling prevents the algorithm from exploiting any canonical description of the graph and is used crucially in our classical lower bounds. The resulting Hamiltonian is specified to the algorithm via an adjacency-list oracle.

\begin{definition}[\sc Hamiltonian associated with a labelled graph]
  \label{Def:mainHam}
  Let $G$ be a graph on $2^m$ vertices with adjacency matrix $A_G$, and let $\pi : V(G) \ra [2^m]$ be a bijection. Define the $m$-qubit Hamiltonian
    \[H_{G,\pi} = - P_{\pi} A_G P_{\pi}^{-1},\]
  where $P_{\pi}$ is the permutation matrix corresponding to $\pi$ (i.e., $(P_{\pi})_{x,u}=1$ if $\pi(u)=x$, and $(P_{\pi})_{x,u}=0$ otherwise).

  The Hamiltonian $H_{G,\pi}$ can be queried via the following adjacency-list oracles:
  \begin{itemize}
    \item {\it (Classical oracle)} Given $x \in [2^m]$, the oracle $\ora_{G,\pi}$ returns the list $\ora_{G,\pi}(x) = \Gamma_x\subset[2^m]$ of labels corresponding to the neighbors of $\pi^{-1}(x)$ in $G$ (equivalently, the non-zero entries in the $x$-th row of $H_{G,\pi}$).
    \item {\it (Quantum oracle)} The oracle implements the unitary operation $\ora_{G,\pi}\ket{x}\ket{0}\mapsto\ket{x}\ket{\Gamma_x}$.
  \end{itemize}
\end{definition}

When clear from the context, we may drop $G$ from the subscript.

\section{Hardness of Ground-State Preparation Problem}
\label{sec:stateprep}
For the remainder of this section, we specialize the Hamiltonian construction of \Cref{Def:mainHam} to the \emph{main graph}~$G_n = G_{n,E_n}$ derived from \Cref{Def:maingraph} with $K=\sqrt{n}$, and we let $A_n = A_{n,E_n}$ denote its adjacency matrix. For $m$ sufficiently large, let $n=\poly(m)$ be such that $G_n$ has~$2^m$ vertices. Let $\pi : V(\main) \ra [2^m]$ be a bijection. We write $\hp:=H_{G_n,\pi}=-P_\pi A_nP_\pi^{-1}$, and denote its ground state by $\psp$. The spectral properties of~$\hp$ follow directly from those of the graph~$G_n$ (\Cref{Lem:spec}). In particular, its ground state~$\psp$ coincides, up to the labeling~$\pi$, with the top eigenvector $\ket{\psi_n}$ of~$A_n$.

\begin{fact}[\sc Spectral properties of $\hp$]
  The Hamiltonian $\hp$ is stoquastic, it has a non-degenerate ground space with spectral gap at least $n/2 - 4\sqrt{2n}$, its ground state is $\psp=P_{\pi}\ket{\psi_n}$.
\end{fact}

Finally, given a guiding state $\psin$ for the unlabelled Hamiltonian $-A_n$, we associate the corresponding family of guiding states $\set{\psinp}_{\pi}$ for the Hamiltonians $\hp$.

\begin{definition}[\sc Guiding state $\psinp$]
  \label{Def:guide}
  Given a guiding state $\psin$ for the Hamiltonian~$-A_n$, we let
  $\psinp=P_{\pi}\psin$ denote the corresponding guiding state for the Hamiltonian $\hp$.
\end{definition}

We now prove our main result (\Cref{thm:lowerbound}): any classical algorithm that makes only a sub-exponential number of queries to the adjacency-list oracle of the Hamiltonian $\hp$, even when provided with sub-exponentially many copies of a guiding state $\psinp$, cannot sample from the ground-state distribution with sufficiently high precision. Before presenting the classical lower bound, we first show that quantum algorithms can solve the problem efficiently.

\begin{proposition}[\sc Quantum algorithm]
  \label{Prop:quantalg}
  There exists a quantum algorithm that solves GGSP on every Hamiltonian $\hp$ (\Cref{Def:mainHam}), using $\poly(n)$ copies of a guiding state~$\psinp$, and $\poly(n)$ quantum queries to the adjacency-list oracle of $\hp$.
\end{proposition}

\begin{proof}
  The graph $G_n$ (and hence the Hamiltonian $\hp$) is $d_{n,1}$-sparse, where $d_{n,1} = \bo{n}$ (\Cref{Fac:grp}). The algorithm follows immediately from \Cref{prop:qalgo}, by noting that simulating a $\bo{n}$-sparse $\poly(n)$-qubit Hamiltonian $H$ requires $T_H = \poly(n)$ quantum gates and queries to its adjacency-list oracle (see, for instance, \cite[Lemma 1.4]{AT07j}).
\end{proof}

The main theorem we prove is that classical algorithms using only sub-exponentially many queries and input samples cannot succeed in solving the GGSP problem on $\hp$ with high probability.

\begin{theorem}[\sc Classical hardness of GGSP]
  \label{thm:lowerbound}
  Let $a \geq 1$ be a constant and $\pt{\psinp}_{\pi}$ a family of guiding states for $\set{\hp}_{\pi}$ (\Cref{Def:guide}) with $\abs{\ip{\psi_{\mathrm{in},\pi}}{\psi_{\pi}}} \geq 1/n^a$ for all $\pi$. Suppose that a classical algorithm is making $q$ queries to the adjacency-list oracle of $\hp$, for a random labeling $\pi$, and is given $t$ independent samples $x_1,\dots,x_t \sim \pinp$ from the distribution $\pinp(x) = \ip{x}{\psi_{\mathrm{in},\pi}}^2$.
  
  If $q + t < 2^{\sqrt{n}}$ then the algorithm cannot (with non-negligible probability over the choice of $\pi$) output a sample $y$ from a distribution $\poutp$ that is the measurement distribution of a refined state~$\psoutp$ satisfying $\abs{\ip{\psi_{\mathrm{out},\pi}}{\psi_{\pi}}} \geq 3/4$, and such that the joint distribution $J_{\pi}$ of $(x_1,\dots,x_t,y)$ is within total variation distance $1/5$ of~$\pinp^{\otimes t} \times \poutp$.
\end{theorem}

The theorem is proved by analyzing a certain localization property of the ground state~$\psp$ of~$\hp$, which must also be satisfied (at least to some extent) by the output state $\psoutp$. It concerns the probability that the output vertex lies at distance at least~$\ell_{n,K}$ from the \emph{closure} of the input vertices. We measure distance using $\dist_E$ (\Cref{Def:loc}). Informally, the closure enlarges the input set until every short path between its vertices lies entirely within the set. Such closure operations also arise, for instance, in LP integrality gaps and sum-of-squares lower bounds~\cite{CMM09c,BCK15c}.

\begin{definition}[\sc Closure of a set of vertices]
\label{def:closure}
    Let $G=(V,E)$ be a finite graph, let $r \geq 2$ be an integer, and let $U\subseteq V$. Starting from $U$, repeatedly apply the following rule for as long as it is applicable:
    \begin{itemize}
      \item[] If there is a path of length $\ell \in \{2,\dots,r\}$ whose two endpoints are distinct vertices of the current set $U$ and whose internal vertices all lie outside $U$, add its internal vertices to $U$.
    \end{itemize}
    Since each step strictly enlarges $U$ and $V$ is finite, the process terminates. The resulting set does not depend on the order in which the steps are performed. We call it the \emph{closure} of $U$ with parameter~$r$ and denote it by $\closure_{G,r}(U)$. We drop the subscript $G$ when the graph is clear from context.
\end{definition}

First, we show in \Cref{sec:hfarideal} that any input-output distribution that correctly solves the GGSP problem satisfies the localization property with probability at least $1/10$.

\begin{rproposition}[\Cref{Prop:hfarideal} ({\sc Input-Output state localization})]
  Fix any labeling $\pi$ and suppose that a classical algorithm solves the GGSP problem on $\hp$ under the conditions of \Cref{thm:lowerbound}. Let~$J_{\pi}$ denote the joint distribution of the input-output samples $(x_1,x_2,\dots,x_t,y)$ of the algorithm under the vertex labeling $\pi$. Then, if $t < 2^{\sqrt{n}}$, the probability that the output vertex $y$ lies at distance at least~$\ell_{n,K}$ from $\br{U} = \closure_{2\ell_{n,K}+3}(\ep{\pi^{-1}\pt{\{x_1,x_2,\dots,x_t\}}})$ satisfies,
    \begin{equation}
      \label{Eq:correct}
      \pr_{(x_1,\dots,x_t,y) \sim J_{\pi}}*{\dist_E\big(\br{U},\pi^{-1}(y)\big) > \ell_{n,K}} \geq 1/10.
    \end{equation}
\end{rproposition}

Conversely, we prove in \Cref{sec:classhard} that for any classical algorithm using at most $2^{\sqrt{n}}$ queries and input samples, the probability of the above event is exponentially small. We stress that this result holds for a random $\pi$, in contrast to the proposition above which is valid for all $\pi$.

\begin{rproposition}[\Cref{Prop:classhard} ({\sc Hardness of exploring the graph $G_n$})]
  Suppose that a classical algorithm is given access to $\hp$ for a random labeling $\pi$, as well as input samples $x_1,\dots,x_t$ from a distribution as specified in \Cref{thm:lowerbound} with $t < 2^{\sqrt{n}}$. Suppose that it outputs an additional sample~$y$, and let~$J_{\pi}$ denote the joint distribution of the input-output samples $(x_1,x_2,\dots,x_t,y)$. If the algorithm makes fewer than $2^{\sqrt{n}}$ queries, then the probability that the output vertex~$y$ lies at distance at least~$\ell_{n,K}$ from $\br{U} = \closure_{2\ell_{n,K}+3}(\ep{\pi^{-1}\pt{\{x_1,x_2,\dots,x_t\}}})$ satisfies,
    \begin{equation}
      \label{Eq:wrong}
      \pr_{\pi,(x_1,\dots,x_t,y) \sim J_{\pi}}*{\dist_E\big(\br{U},\pi^{-1}(y)\big) > \ell_{n,K}} \leq 2^{-n\log{n}}.
    \end{equation}
\end{rproposition}

We can now conclude the proof of the main theorem.

\begin{proof}[Proof of \Cref{thm:lowerbound}]
  Fix $\psin$ to be any guiding state for the Hamiltonian $-A_n$. Consider a classical algorithm that solves the GGSP problem on every Hamiltonian $\hp$, using the guiding state $\psinp = P_{\pi} \psin$ as specified in \Cref{Def:guide} and \Cref{thm:lowerbound}. Suppose, for the sake of contradiction, that it makes fewer than $2^{\sqrt{n}}$ queries. Then, by \Cref{Prop:classhard}, for an overwhelming fraction of labelings $\pi$, its input-output distribution $J_{\pi}$ violates \Cref{Eq:correct}. It implies by \Cref{Prop:hfarideal} that the algorithm cannot be correct on~$\hp$ with non-negligible probability over the choice of $\pi$.
\end{proof}

  \subsection{Proof of ground-state localization properties}
  \label{sec:hfarideal}
  In this section, we lower bound the probability that a successful algorithm solving the GGSP problem samples a vertex at a large distance from its input. First, we prove the statement under the exact ground-state distribution.

\begin{lemma}[\sc Ground-state localization]
  \label{Lem:hfaridealpair}
  Let $U \subseteq V$ be any subset of vertices of $G_n$. For any integer $\ell$, the probability that a vertex sampled from the ground-state distribution $p_n(v) = \ip{v}{\psi_n}^2$ of~$-A_n$ lies at distance at least~$\ell$ from $U$ is at least,
    \[\pr_{v \sim p_n}*{\dist_E(U,v) \geq \ell} \geq 1 - \frac{\abs{U} \cdot n^{\ell}}{N_E}\]
  where $N_E$ is the size of the expander $E_n$ in \Cref{Prop:expander}. 
\end{lemma}

\begin{proof}
  First, we prove the lemma when $U = \set{u}$ is of size $1$.
  Let $p_E$ be the conditional distribution of~$p_n$ on the set of expander vertices $V_E$. By \Cref{Lem:spec}, when $v$ is sampled according to~$p_n$, the corresponding expander vertex $\ep{v}$ follows the distribution $p_E$. Hence,
    \begin{equation}
      \label{Eq:marg}
      \pr_{v \sim p_n}*{\dist_E(u,v) < \ell} = \pr_{v \sim p_E}*{\dist_E(u,v) < \ell}.
    \end{equation}
  By \Cref{Lem:spec}, $p_E$ is the uniform distribution over $V_E$. Moreover, by construction, the expander graph over $V_E$ has degree $d_n = n$. Hence, the number of expander vertices $v$ that are at distance less than $\ell$ from $\ep{u}$ is at most $1 + n + n^2 + \dots + n^{\ell-1} \leq n^{\ell}$. Consequently, the probability of sampling $v$ at distance $< \ell$ from $\ep{u}$, under the distribution $p_E$, can be upper bounded by the ratio between the number of expander vertices at distance less than $\ell$ from $\ep{u}$ and the total number $N_E$ of expander vertices:
    \begin{equation}
      \label{Eq:ball}
      \pr_{v \sim p_E}*{\dist_E(u,v) < \ell} \leq \frac{n^{\ell}}{N_E}.
    \end{equation}
  \Cref{Eq:marg,Eq:ball} imply that $\pr_{v \sim p_n}*{\dist_E(u,v) \geq \ell} \geq 1 - \frac{n^{\ell}}{N_E}$.
  
  When $U$ is of arbitrary size, the lemma follows by applying the above result together with a union bound over the $\abs{U}$ vertices $u \in U$.
\end{proof}

The localization property we use is stated for the closure of a set of vertices (\Cref{def:closure}). As a first step, we show that when the girth is large enough, the closure of a set $U$ with $\abs{U}=t$ and parameter $r$ has size linear in $t$ and $r$. A similar bound was shown in~\cite[Lemma 4.4]{BCK15c} under different conditions.

\begin{lemma}[\sc Size of closure]\label{lem:closure-size}
  Let $U \subseteq V$ be a set of vertices such that $\abs{U} = t \geq 1$ and~$r \geq 2$ be an integer. Suppose $G$ has girth at least $g \geq 4r\log(2rt) + 4$. Then the procedure of \Cref{def:closure} makes fewer than $2t$ additions and the final number of vertices is $\abs{\closure_r(U)} < 2rt$.
\end{lemma}

\begin{proof}
  Let $U_i$ be the set after $i$ steps, and let $G[U_i]$ denote the subgraph of $G$ induced by $U_i$. We write $v_i=\abs{U_i}$ and $e_i$ for the number of vertices and edges of $G[U_i]$, respectively. A step adds a path of length $\ell\le r$, hence $\ell-1\le r-1$ new vertices and at least $\ell$ new edges (each path edge touches a new vertex). So $v_i\le t+i(r-1)$ and $e_i-v_i\ge i-t$.

  Suppose the process makes at least $2t$ steps. Then $v_{2t}<2rt$ and $e_{2t}-v_{2t}\ge t$, so $G[U_{2t}]$ has average degree $d\ge 2+2t/v_{2t}>2+1/r$. Being a subgraph of $G$, it also has girth at least $g$. Hence the average-degree Moore bound~\cite{AHL02j} and $1+x\ge\exp\big(\tfrac{x}{1+x}\big)$ with $x=1/r$ give
  \[
      v_{2t}\ge (d-1)^{\lfloor (g-1)/2\rfloor}\ge \Bigl(1+\tfrac1r\Bigr)^{2r\log(2rt)}\ge \exp\Bigl(\tfrac{2r\log(2rt)}{r+1}\Bigr)\ge 2rt,
  \]
  where we used $\lfloor (g-1)/2\rfloor\ge 2r\log(2rt)$ and $r+1\le 2r$. This contradicts $v_{2t}<2rt$.

  Hence fewer than $2t$ steps are made, and $\abs{\closure_r(U)}\le t+(2t-1)(r-1)<2rt$.
\end{proof}

Combining \Cref{Lem:hfaridealpair} with the size bound on the closure gives the version of the localization bound that we use below: a sampled vertex is far from the closure of $U$, not just from $U$ itself. This requires the girth of $E_n$ to be large enough for the size bound to apply, which holds for the parameters we use by \Cref{Prop:expander}.

\begin{lemma}
\label{lem:hfairidealloc}
  Let $U \subseteq V$ be a set of vertices of $G_n$ of size $t <2^{\sqrt{n}}$. Then the probability that a vertex sampled from the ground-state distribution
  $p_n(v) = \ip{v}{\psi_n}^2$ of~$-A_n$ lies at distance at least~$\ell_{n,K}$ from $\br{U} = \closure_{2\ell_{n,K}+3}(\ep{U})$ is at least
  \[
    \pr_{v \sim p_n}*{\dist_E(\br{U},v) > \ell_{n,K}}
    \geq 1 - 2^{-n^2}.
  \]
\end{lemma}
\begin{proof}
  Set $r = 2\ell_{n,K}+3$ and $t = \abs{U}$. By \Cref{Prop:expander}, the girth
  $g_n$ of $E_n$ satisfies $g_n > 4r\log(2rt) + 4$, so \Cref{lem:closure-size}
  applies and gives $|\br{U}| < 2rt = 2(2\ell_{n,K}+3)t$. Applying
  \Cref{Lem:hfaridealpair} to the set $\br{U}$, we obtain
  \[
    \pr_{v \sim p_n}*{\dist_E(\br{U},v) \geq \ell_{n,K}+1}
    \geq 1 - \frac{|\br{U}|\, n^{\ell_{n,K}+1}}{N_E}
    \geq 1 - \frac{2(2\ell_{n,K}+3)t\, n^{\ell_{n,K}+1}}{N_E}.
  \]
  Recall that $\ell_{n,K}=10n^2\log n$ and $N_E \geq 2^{100 n^{5/2} \log^2 n}$. Since $t<2^{\sqrt n}$, for all sufficiently large $n$,
  \[
    \frac{2(2\ell_{n,K}+3)\,t\, n^{\ell_{n,K}+1}}{N_E} \le 2^{-n^2}. \qedhere
  \]
\end{proof}

The previous lemma analyzes the ideal case in which the output vertex is sampled according to the \emph{ideal} ground-state distribution $p_n$. However, the GGSP problem allows outputting vertices from a perturbed distribution $\pout$, which may adversarially affect the occurrence of the above event. The next result shows that the event still occurs in this setting, at least with constant probability, when $\abs{U} = t$ is of size at most $2^{\sqrt{n}}$ and $\ell$ is sufficiently smaller than the size of the expander.

\begin{corollary}
  \label{Cor:hfarideal}
  Consider a joint distribution $J$ over $t+1$ vertices $(U=\{u_1,u_2,\dots,u_t\},v)$ of $G_n$. Suppose that $J$ is within total variation distance $1/5$ of a product distribution $\pin^{\otimes t} \times \pout$, where $F(\pout,p_n) \ge (3/4)^2$. Then, for $n$ large enough and $t < 2^{\sqrt{n}}$, the probability that the vertex~$v$ lies at distance at least $\ell_{n,K}$ from $\br{U} = \closure_{2\ell_{n,K}+3}(\ep{U})$ satisfies,
    \[\pr_{(u_1,\dots,u_t,v) \sim J}*{\dist_E(\br{U},v) > \ell_{n,K}} \geq 1/10.\]
\end{corollary}

\begin{proof}
  The result is shown by upper-bounding the total variation distance between the distributions $\pin^{\otimes t}\times p_n$ and $J$, and then applying the result from \Cref{lem:hfairidealloc}.

  The total variation distance is bounded as follows,
    \begin{align*}
      \mathrm{TV}(\pin^{\otimes t}\times p_n, J)
      & \leq \mathrm{TV}(\pin^{\otimes t}\times p_n,\pin^{\otimes t}\times \pout) + \mathrm{TV}(\pin^{\otimes t}\times \pout, J) \tag{by the triangle inequality} \\ 
      & \leq \sqrt{1 - F(\pout,p_n)} + \mathrm{TV}(\pin^{\otimes t}\times \pout, J) \tag{by relating the TV distance and the fidelity} \\
      & \leq \sqrt{1-\frac{9}{16}} + \frac{1}{5} \leq 0.87. \tag{by the assumptions}
    \end{align*}
  
  Hence, with $r = 2\ell_{n,K}+3$, the probability under~$J$ of sampling a vertex at distance at least $\ell_{n,K}$ from $\closure_r(\ep{U})$ is at least
    \begin{align*}
      & \pr_{(u_1,\dots,u_t,v) \sim J}*{\dist_E(\br{U},v) > \ell_{n,K}} \\
       & \qquad\qquad  \geq \pr_{(u_1,\dots,u_t,v) \sim \pin^{\otimes t}\times p_n}*{\dist_E(\br{U},v) > \ell_{n,K}} - \mathrm{TV}(\pin^{\otimes t}\times p_n,J) \tag{by definition of the TV distance} \\
       & \qquad\qquad  \geq \min_{W \subseteq V : \abs{W} \leq t} \pr_{v \sim p_n}*{\dist_E(\closure_{r}(\ep{W}),v) > \ell_{n,K}} - \mathrm{TV}(\pin^{\otimes t}\times p_n,J) \tag{by the independence of the distributions} \\
       & \qquad\qquad  \geq 1 - 2^{-n^2} - 0.87 \tag{by \Cref{lem:hfairidealloc}}\\
       & \qquad\qquad \geq 1/10. \tag{for $n$ large enough}
    \end{align*}
\end{proof}

As a direct corollary, the same result holds for the input-output samples of any algorithm solving the GGSP problem, since the distributions $J$ and $J_{\pi}$ are the same object up to a relabeling of the vertices.

\begin{proposition}[\sc Input-Output state localization]
  \label{Prop:hfarideal}
  Fix any labeling $\pi$ and suppose that a classical algorithm solves the GGSP problem on $\hp$ under the conditions of \Cref{thm:lowerbound}. Let~$J_{\pi}$ denote the joint distribution of the input-output samples $(x_1,x_2,\dots,x_t,y)$ of the algorithm under the vertex labeling $\pi$. Then, if $t < 2^{\sqrt{n}}$, the probability that the output vertex $y$ lies at distance at least~$\ell_{n,K}$ from $\br{U} = \closure_{2\ell_{n,K}+3}(\ep{\pi^{-1}\pt{\{x_1,x_2,\dots,x_t\}}})$, satisfies,
    \[
      \pr_{(x_1,\dots,x_t,y) \sim J_{\pi}}*{\dist_E\big(\br{U},\pi^{-1}(y)\big) > \ell_{n,K}} \geq 1/10.
    \]
\end{proposition}

  \subsection{Proof of hardness of GGSP}
  \label{sec:classhard}
  This section proves a converse result to the previous one. Namely, it demonstrates that \emph{no} classical algorithm can output a vertex far from the input vertices unless it makes at least sub-exponentially many queries to $\hp$. For simplicity, we consider throughout that the output is always a non-isolated vertex since the ground state is not supported on isolated vertices.

First, we show that any classical algorithm operating on the graph $G_n$ is forced to explore the graph by querying connected components around its input vertices. This is similar to the argument in~\cite{CCD03c}.

\begin{lemma}[\sc Local exploration]
  \label{lem:conncomp}
  Let $U \subseteq V$ be any subset of vertices of $G_n$ with $\abs{U} \leq \abs{V}/2$. 
  Suppose that a classical algorithm is given access to the adjacency-list oracle $\ora_{\pi}$, for a random~$\pi$, together with $\pi(U)$, and that it makes $q$ queries. Then, except with probability at most $q \cdot 2^{-2n^{3}\log^2{n}}$, the set of all vertices queried by the algorithm at any point during its execution induces a collection of connected components rooted at the vertices in~$U$, together with some isolated vertices in $V_I$.
\end{lemma}

\begin{proof}
  This is because the set of isolated vertices $V_I$ in \Cref{Def:maingraph} constitutes an overwhelming majority of the vertices of the graph $G_n$. Hence, if an algorithm attempts to query a vertex that is not connected to any of its previously queried vertices, then, due to the randomness of the labeling~$\pi$, it has overwhelming probability of querying an isolated vertex.

  In more detail, by \Cref{Def:maingraph}, the ratio between the number of non-isolated vertices and the total number of (non-input) vertices in $G_n$ is at most $(\abs{V_T} + \abs{V_E})/(\abs{V}-\abs{U}) \leq 2 (\abs{V_T} + \abs{V_E})/\abs{V} = 2 N_E (1 + \sqrt{n} \sum_{k=1}^{K-1}|\treh_{n,k}|) / 2^{15n^{3}\log^2{n}} \leq 2^{-2n^{3}\log^2{n}}$. Hence, by a union bound over the $q$ queries made by an algorithm, the probability that it queries a non-isolated vertex that is not connected to any of its previously explored vertices is at most~$q \cdot 2^{-2n^{3}\log^2{n}}$.
\end{proof}

We now analyze how local graph exploration around an input set of vertices $U$ in $G_n$ behaves. The case $\abs{U} = 1$ was addressed in \Cref{fig:overview,Fig:expanderlocal} and \Cref{Lem:locss}, where we showed that, within truncated neighborhood $B_{\ell_{n,K}}(u)$ of $u$, the expander edges outgoing from $u$ form multiple copies of the tree~$\treh_{n,K}$. The next lemma treats the case~$\abs{U} > 1$. The analysis is more involved, as different connected components may intersect and cycles may appear in the exploration graph. We show that such events must be confined to the closure of $\ep{U}$, a small superset~$\br{U} \supseteq \ep{U}$ of the input vertices, as defined in \Cref{def:closure}. The remainder of the local exploration departing from these vertices can then be described by the same trees as before. 
This is illustrated in \Cref{fig:zoom}.

\begin{figure}[ht!]
  \begin{center}
  \begin{tikzpicture}[scale=0.75]
  \draw[line width=1,dashed,rounded corners] (1,2.3) -- (1,4) -- (1.7,4); 
  \draw[line width=1,dashed,rounded corners] (2.7,1) -- (2.3,1) -- (2.3,2) -- (7,2) -- (7,2.7);
  \draw[line width=1,dashed,rounded corners] (9.3,1) -- (10,1) -- (10,2) -- (9,2) -- (9,2.7);
  \draw[line width=1,dashed,rounded corners] (6.3,6) -- (11.7,6);
  \draw[line width=1,dashed,rounded corners] (5,9.3) -- (5,10) -- (17,10) -- (17,9.3);
  \draw[line width=1] (5,6) -- (3,6) -- (3,4) -- (3,10); 
  \draw[line width=1] (5,4) rectangle (13,8);
  \draw[line width=1] (13,8) rectangle (17,0);
  \draw[line width=1] (4,1) -- (8,1);
  \foreach \x/\y/\angle/\anchor in {5/7.75/0/south,7/7.75/0/south,9/7.75/0/south,11/7.75/0/south,13/7.75/0/south,15/7.75/0/south,17/7.75/0/south,%
    5.67/6.7/270/west,%
    12.35/6.7/90/east,%
    5/2.4/180/north,7/2.4/180/north,9/2.4/180/north,11/2.4/180/north,%
    13/-1.6/180/north,15/-1.6/180/north,17/-1.6/180/north,%
    12.35/2.7/90/east,%
    13.65/4.7/270/west,16.35/2.7/90/east,16.35/4.7/90/east
  }{
    \draw node[anchor=\anchor,rotate=\angle] at (\x,\y) {\mytree{white!50!black}};
    }
  \foreach \x/\y/\angle/\anchor in {16.35/6.7/90/east,2.35/6.7/90/east,2.35/4.7/90/east,2.35/8.7/90/east,2.35/10.7/90/east,1/0.75/0/south,3.35/1.7/90/east,8.7/1.7/270/west,6/-0.6/180/north}{
    \draw node[anchor=\anchor,rotate=\angle] at (\x,\y) {\mytree{white!50!black}};
  }
  \foreach \x/\y in {1/1,4/1,8/1,3/4,3/10,5/4,11/8,13/4,13/0,17/2,17/8}{
    \draw[draw=orange,opacity=0.5,line width=2] (\x+0.2,\y) -- (\x,\y+0.2) -- (\x-0.2,\y) -- (\x,\y-0.2) -- cycle;
  }
\end{tikzpicture}
  \caption{Pictorial representation of the local structure of the graph $G_n$ (only expander vertices are shown). The orange diamonds represent the vertices in the set $\ep{U}$. The set~$\br{U} = \closure_{2\ell_{n,K} + 3}(\ep{U})$, shown by the greyed vertices, is the fixed point of the following process: as long as two vertices of the current set are joined by a path of length at most $2\ell_{n,K}+3$ whose internal vertices lie outside the set, those internal vertices are added to it. The solid edges are such connecting paths, including paths between vertices added in earlier rounds. The triangles depict the local tree structure of the truncated neighborhood $B_{\ell_{n,K}}(\br{U})$ of $\br{U}$. The white circles represent the vertices identified with the $0$-level leaves of $\treh_{n,K}$. The dashed lines indicate other paths in the expander graph that may create cycles but lie too far from $\br{U}$ to be detected by a local exploration.}
  \label{fig:zoom}
\end{center}
\end{figure}

\begin{lemma}[\sc Local structure of $G_n$ around multiple roots]
  \label{Lem:new-hardmaing}
  Let $U_0 \subseteq V$ be any subset of \emph{expander} vertices of $G_n$. Let $\br{U} = \closure_{2\ell_{n,K}+3}(U_0)$ denote the closure of $U_0$ with parameter $2\ell_{n,K}+3$. Consider the truncated neighborhood $B_{\ell_{n,K}}(\br{U}) \subseteq G_n$ of $\br{U}$ as defined in \Cref{def:trunc-neigh}. Then $B_{\ell_{n,K}}(\br{U})$ is isomorphic to the graph obtained by taking the graph induced by $\br{U}$ in $G_n$ and then attaching to each $u \in \br{U}$ as many copies of $\treh_{n,k}$ as in~$G_n$ (for $k = 1,\dots, K-1$), and as many copies of~$\treh_{n,K}$ as the number of expander neighbors of $u$ outside $\br{U}$.

  Moreover, the isomorphism can be chosen so that the expander vertices at distance exactly $\ell_{n,K} + 1$ from $\br{U}$ in $G_n$ map to the $0$-level leaves of the attached copies of~$\treh_{n,K}$.
\end{lemma}

\begin{proof}
  Let $R := \ell_{n,K}+1$. Fix $u \in \br{U}$ and an expander neighbor $v \in V_E \setminus \br{U}$ of~$u$, and let $T(u,v)$ be the connected component of $v$ in $B_{\ell_{n,K}}(\br{U})$ \emph{after deleting} the edge~$(u,v)$. We show that the subgraphs $\set{T(u,v)}_{u,v}$ are disjoint trees, each isomorphic to $\treh_{n,K}$. Take a breadth-first forest of the expander vertices of $B_{\ell_{n,K}}(\br{U})$ rooted at~$\br{U}$; by \Cref{def:trunc-neigh}, its trees have depth at most~$R$. Any expander edge of $B_{\ell_{n,K}}(\br{U})$ not in this forest and not having both endpoints in $\br{U}$ either joins two vertices of the same tree, closing a cycle of length at most $2R+1$, or joins two different trees, giving a path of length between $2$ and $2R+1 = 2\ell_{n,K}+3$ between distinct vertices of $\br{U}$ with all internal vertices outside $\br{U}$. The first case contradicts the girth bound $g_n > 4\ell_{n,K}+8 > 2R+1$, and the second contradicts $\br{U}$ being closed. Hence every such edge lies in the forest, and since decoration edges only attach trees, no cycles arise and the $T(u,v)$ are disjoint trees. 
  
  As in \Cref{Lem:locss}, the expander vertices of $T(u,v)$ form a full tree of depth $\ell_{n,K}$ and degree~$d_{n,K}$, to which the decorations of \Cref{Def:maingraph} are attached. This matches \Cref{Def:selfsim} and \Cref{Fac:botup}, so $T(u,v)$ is isomorphic to $\treh_{n,K}$.
\end{proof}

The next result studies the hardness of exploring the tree $\treh_{n,K}$, starting from its root and attempting to reach a $0$-level leaf.

\begin{lemma}[\sc Hardness of exploring self-similar trees]
  \label{Lem:hardss}
  Let $K = \sqrt{n}$, as in \Cref{Def:maingraph}. Suppose a classical algorithm explores the tree $\treh_{n,K}$, starting at its root and querying only vertices adjacent to those it has already visited. If the algorithm makes fewer than~$2^{\sqrt{n}}$ queries, the probability that it ever reaches a $0$-level leaf of $\treh_{n,K}$ is at most $2^{-2n\log{n}}$.
\end{lemma}

\begin{proof}
  The proof essentially follows \cite{GHV21c}. We provide a detailed version in \Cref{App:trees}, adapted to our choice of parameters and to our definition of self-similar trees. The lemma is reformulated there as \Cref{Prop:exithard}.
\end{proof}

By combining the three lemmas above, we show that exploring the main graph $G_n$ to reach expander vertices far from a given set $U$ of input vertices is hard.

\begin{proposition}[\sc Hardness of exploring the graph $G_n$]
  \label{Prop:classhard}
  Suppose that a classical algorithm is given access to $\hp$ for a random labeling $\pi$, as well as input samples $x_1,\dots,x_t$ from a distribution as specified in \Cref{thm:lowerbound} with $t < 2^{\sqrt{n}}$. Suppose that it outputs an additional sample~$y$, and let~$J_{\pi}$ denote the joint distribution of the input-output samples $(x_1,x_2,\dots,x_t,y)$. If the algorithm makes fewer than $2^{\sqrt{n}}$ queries, then the probability that the output vertex~$y$ lies at distance at least~$\ell_{n,K}$ from $\br{U} = \closure_{2\ell_{n,K}+3}(\ep{\pi^{-1}\pt{\{x_1,x_2,\dots,x_t\}}})$ satisfies,
    \[\pr_{\pi,(x_1,\dots,x_t,y) \sim J_{\pi}}*{\dist_E\big(\br{U},\pi^{-1}(y)\big) > \ell_{n,K}} \leq 2^{-n\log{n}}.\] 
\end{proposition}

\begin{proof}
  Let $U = \pi^{-1}\pt{\set{x_1,\dots,x_t}}$ denote the (random) set of input vertices sampled under~$J_{\pi}$. By \Cref{lem:conncomp}, except with probability $2^{\sqrt{n}} \cdot 2^{-2n^{3}\log^2{n}} < 2^{-2n\log{n}}$, we may assume that the consecutive queries of the algorithm form a set of connected components in $G_n$ rooted at the vertices in $U$.

  Let $v = \pi^{-1}(y)$ be the vertex in $G_n$ output by the algorithm. By \Cref{Lem:new-hardmaing}, if the algorithm were to explore a vertex at distance at least $\ell_{n,K}+1$ from $\br{U}$, then there must exist an exploration tree attached to one of the vertices in $\br{U}$ that can be embedded into~$\treh_{n,K}$ and reaches one of its $0$-level leaves. Fix one such copy $T$. Although the algorithm's queries to $T$ may be interleaved adaptively with its queries elsewhere, we may simulate every query of the algorithm outside~$T$ using the rest of the graph together with the given input labels, which expose no internal vertex of~$T$, and only charge to~$T$ the at most $2^{\sqrt n}$ queries that $\mcl{A}$ makes inside it. Applying \Cref{Lem:hardss} to this restricted exploration bounds the probability that the algorithm reaches a $0$-level leaf of~$T$ by $2^{-2n\log n}$, independently of what it learned from the other copies. By \Cref{lem:closure-size}, the number of possible roots to consider for the embedding is at most $\abs{\br{U}} < 2n(2\ell_{n,K}+3)t$. A union bound then implies that the overall success probability is at most $2n(2\ell_{n,K}+3)t \cdot 2^{-2n\log{n}} \leq 2^{-\frac{3}{2}n\log{n}}$, for $n$ large enough.
\end{proof}

\section{Hardness of Ground-State Energy Estimation Problem}
  \label{sec:energyhard}
  In this section, we prove the classical hardness of GGSE by studying two families of stoquastic Hamiltonians obtained from two closely related graphs in \Cref{Def:maingraph}. We fix $K = \sqrt{n}$ in this section. The first family corresponds to the main graph~$G_n = G_{n,E_n}$ (studied already in \Cref{sec:stateprep}), which gives the $\yes$ instances. The second family corresponds to the \emph{envelope graph} $F_n = G_{n,E'_n}$, and gives the $\no$ instances. Solving GGSE on these instances amounts to distinguishing the two graph families by the next lemma.

\begin{lemma}[\sc GGSE instances]
  \label{lem:ggse-instance}
  There exist thresholds $a<b$ with $b-a=\Omega(\sqrt n)$ such that, for every labeling $\pi$, $H_{\main,\pi}$ satisfies the $\yes$ promise of GGSE (i.e., $\lambda_{\rm min}(H_{\main,\pi}) \leq a$) and $H_{\envl,\pi}$ satisfies the $\no$ promise (i.e., $\lambda_{\rm min}(H_{\envl,\pi}) \geq b$).
\end{lemma}

\begin{proof}
  Let $\lambda_{\rm max}(\main)$ and $\lambda_{\rm max}(\envl)$ denote the top eigenvalues of $\main$ and $\envl$ respectively. Then, by \Cref{App:app-spectral}, there exists an absolute constant $C>0$ such that $\lambda_{\rm max}(\main)\geq n$ and $\lambda_{\rm max}(\envl)\leq n-\sqrt n+C$. Set $a:=-n$, $b:=-n+\sqrt n-C$. Then $\lambda_{\min}(H_{\main,\pi})\leq a$, $\lambda_{\min}(H_{\envl,\pi})\geq b$, and $b-a=\sqrt n-C=\Omega(\sqrt n)$.
\end{proof}

As in the GGSP setting, the corresponding quantum problem is easy.

\begin{proposition}[\sc Quantum algorithm]\label{Prop:q-algo-ggse}
  There exists a quantum algorithm that solves GGSE on every Hamiltonian $H_{G,\pi}$, for $G\in \{\main,\envl\}$, using $\poly(n)$ copies of a guiding state~$\psinp$, and $\poly(n)$ quantum queries to the adjacency-list oracle of $H_{G,\pi}$.
\end{proposition}

\begin{proof}
  Both $\main$ and $\envl$, and hence their associated Hamiltonians, are $\bo{n}$-sparse. The result follows from \Cref{prop:q-ggse}, since a $\bo{n}$-sparse $\poly(n)$-qubit Hamiltonian can be simulated using $\poly(n)$ quantum gates and queries to its adjacency-list oracle.
\end{proof}

The main and envelope graphs differ in the degree of their base expanders and in the decorations attached to each expander vertex. These changes compensate for one another, so every internal vertex has degree $d_{n,1} = 2n-\sqrt{n}$ in both graphs. The constructions also have the same total number of vertices and the same girth lower bound. Thus these elementary graph parameters do not distinguish the two families. The main result of this section is to show that, in fact, no classical algorithm using only sub-exponentially many queries and input samples can distinguish between the two graphs (and hence solve GGSE).

\begin{theorem}[\sc Classical hardness of GGSE]
  \label{Thm:GEEHard}
  Suppose that a classical algorithm is making~$q$ queries to the adjacency-list oracle of $H_{G,\pi}$, for a random $G \in \{\main,\envl\}$ and a random labeling~$\pi$, and is given $t$ independent samples $x_1,\dots,x_t \sim p_{E,\pi}$ from the distribution $p_{E,\pi}(x) = \ip{x}{\psi_{E,\pi}}^2$, where~$\ket{\psi_{E,\pi}}$ is the uniform superposition over the expander vertices of $G$. If $q + t < 2^{\sqrt{n}}$ then the algorithm cannot solve the GGSE problem on $G$ with success probability greater than~$1/2 + 2^{-n}$. The guiding state has overlap $\abs{\ip{\psi_{E,\pi}}{\psi_{G,\pi}}} = 1-O(1/n)$ with the ground state $\ket{\psi_{G,\pi}}$ of $H_{G,\pi}$.
\end{theorem}

  \subsection{Local graph structures}
  This section studies the local structures around expander vertices in $\main$ and $\envl$. Although these structures differ, we show that the local structure in $\main$ embeds into that in $\envl$ up to a natural frontier defined by a subset of the 0-level leaves of $\treh_{n,K}$.

By \Cref{Lem:locss} and \Cref{Fig:expanderlocal}, the truncated neighborhood of an expander vertex $u$ differs between the two graphs in that the expander-generated branches are connected to copies of~$\treh_{n,K}$ in $\main$, but to copies of~$\treh_{n,K+1}$ in $\envl$. We relate the two local structures by first showing that~$\treh_{n,K}$ can itself be embedded into a copy of $\treh_{n,K+1}$. The fact that $\treh_{n,K}$ is a subgraph of $\treh_{n,K+1}$ is straightforward, since the latter graph contains multiple copies of the former. The much subtler property, which we demonstrate next, is that such an embedding is possible even when the roots of the two graphs are matched together. This is illustrated in \Cref{Fig:self-sim-emb}.

\begin{figure}[ht!]
  \begin{center}
    \input{Figures/self-sim-emb.tex}
   \caption{
      Schematic illustration of the rooted embedding of $\treh_{n,K}$ into the truncated copy of~$\treh_{n,K+1}$. The zoomed views show one recursive step at corresponding $0$-level internal vertices. The attached copies of $\treh_{n,k}$, for $k=1,\ldots,K-1$, are mapped by canonical rooted isomorphisms (red regions). Of the remaining branches, $n-\sqrt{n}-1$ continue recursively through $0$-level children (green regions). For each of the other $\sqrt n$ children, shown in the yellow region on the left, the entire subtree rooted at that child is mapped onto one of the copies of $\treh_{n,K}$ attached as a decoration in the red region on the right.}
  \label{Fig:self-sim-emb}
  \end{center}
\end{figure}

\begin{lemma}[\sc Embedding of self-similar trees]
  \label{lem:embed-ss}
  For each $1 \leq k \leq \sqrt{n}$, there exists an embedding of $\treh_{n,k}$ into $\treh_{n,k+1}$ that maps the roots to each other and preserves the degrees of the nodes, except at some of the $0$-level leaves of $\treh_{n,k}$.
\end{lemma}

\begin{proof}
  By \Cref{Fac:botup}, $\treh_{n,k}$ and $\treh_{n,k+1}$ are obtained by decorating the perfect trees $C_{n,k}$ and~$C_{n,k+1}$ respectively (blue edges in \Cref{Fig:self-sim-emb}). The tree~$C_{n,k+1}$ is deeper but has smaller degree $d_{n,k+1} = d_{n,k} - \sqrt{n}$. On the other hand, each node carries $\sqrt{n}$ additional decorations in $\treh_{n,k+1}$, each a copy of~$\treh_{n,k}$.

  We construct the embedding in two steps. First, consider a perfect subtree of $C_{n,k}$ of degree~$d_{n,k+1}$ and depth $\ell_{n,k}$, together with all decorations attached to it. The perfect subtree embeds into the first $\ell_{n,k}$ levels of $C_{n,k+1}$ (blue edges in green region in \Cref{Fig:self-sim-emb}), and its decorations embed into those of~$C_{n,k+1}$ since they form a subset of the latter (red region in \Cref{Fig:self-sim-emb}).

  Next, it remains to embed, for each node of $C_{n,k}$, the $d_{n,k} - d_{n,k+1} = \sqrt{n}$ subtrees rooted at 0-level children not yet accounted for (yellow region in \Cref{Fig:self-sim-emb}). Consider one such child, at some depth $1 \leq \ell \leq \ell_{n,k}$. The subtree rooted at it is isomorphic to $C_{n,k}$ truncated at depth $\ell_{n,k} - \ell$ and decorated as in $\treh_{n,k}$. Each such subtree therefore embeds into one of the $\sqrt{n}$ additional copies of $\treh_{n,k}$ attached at the corresponding node of~$C_{n,k+1}$. 
  
  Because of the depth truncated at level $\ell$, the $0$-level leaves of the subtree are mapped to internal nodes of the copy of $\treh_{n,k}$. These are the only vertices whose degree is not preserved by the embedding.
\end{proof}

The above result extends readily to an embedding of the local structure of $\main$ into that of~$\envl$, as we demonstrate next and illustrate in \Cref{Fig:partiso}.
\begin{figure}[ht!]
  \begin{center}

\begin{tikzpicture}[
every node/.style={font=\small},
  centre/.style={
  circle,
  draw=violet!75!black,
  thick,
  fill=magenta!35!white,
  text=violet!75!black,
  minimum size=7mm,
  inner sep=0pt
  },
  common/.style={red!75!black,thick},
  expander/.style={green!45!black,thick},
  extra/.style={blue!75!black,thick}
]

%
\newcommand{\outtree}[6]{%
  \pgfmathsetmacro{\apexdistance}{1.05}
  \pgfmathsetmacro{\basedistance}{1.05+#5}
  \pgfmathsetmacro{\labeldistance}{1.05+0.58*#5}

  \coordinate (treeapex)
    at ($(#1.center)+(#2:\apexdistance)$);

  \coordinate (treebase)
    at ($(#1.center)+(#2:\basedistance)$);

  \coordinate (treebaseup)
    at ($(treebase)+({#2+90}:#6)$);

  \coordinate (treebasedown)
    at ($(treebase)+({#2-90}:#6)$);

  \draw[#3]
    (#1.#2)
    --
    (treeapex);

  \draw[#3]
    (treeapex)
    --
    (treebaseup)
    --
    (treebasedown)
    --
    cycle;

  \pgfmathtruncatemacro{\labelflip}
    {ifthenelse(cos(#2)<0,180,0)}

  \node[
    rotate={#2+\labelflip},
    #3,
    font=\scriptsize
  ] at ($(#1.center)+(#2:\labeldistance)$) {#4};
}

%
\newcommand{\outtreepartial}[9]{%
  \pgfmathsetmacro{\apexdistance}{1.05}
  \pgfmathsetmacro{\outerbasedistance}{1.05+#6}
  \pgfmathsetmacro{\innerbasedistance}{1.05+#8}
  \pgfmathsetmacro{\outerlabeldistance}{1.05+0.78*#6}
  \pgfmathsetmacro{\innerlabeldistance}{1.05+0.52*#8}

  \coordinate (outerapex)
    at ($(#1.center)+(#2:\apexdistance)$);

  \coordinate (outerbase)
    at ($(#1.center)+(#2:\outerbasedistance)$);

  \coordinate (outerbaseup)
    at ($(outerbase)+({#2+90}:#7)$);

  \coordinate (outerbasedown)
    at ($(outerbase)+({#2-90}:#7)$);

  \coordinate (innerbase)
    at ($(#1.center)+(#2:\innerbasedistance)$);

  \coordinate (innerbaseup)
    at ($(innerbase)+({#2+120}:#9)$);

  \coordinate (innerbasedown)
    at ($(innerbase)+({#2-40}:{2.8*#9})$);

  \draw[#3,blue!75!black]
    (#1.#2)
    --
    (outerapex);

  \draw[#3, dotted]
    (outerapex)
    --
    (outerbaseup)
    --
    (outerbasedown)
    --
    cycle;

  \draw[blue!75!black,thick]
    (outerapex)
    --
    (innerbaseup)
    --
    (innerbasedown)
    --
    cycle;

  \pgfmathtruncatemacro{\labelflip}
    {ifthenelse(cos(#2)<0,180,0)}

  \node[
    rotate={#2+\labelflip},
    #3,
    font=\scriptsize
  ] at ($(#1.center)+(#2:\outerlabeldistance-0.35)$)
    {#4};

  \node[
    rotate={#2+\labelflip},
    text=blue!75!black,
    font=\scriptsize,
    inner sep=0pt
  ] at ($(#1.center)+(#2:\innerlabeldistance+0.25)$)
    {#5};
}

%
\newcommand{\branchcount}[6]{%
  \begin{scope}[shift={(#1.center)}]
    \draw[#6]
      (#2:#4)
      arc[
        start angle=#2,
        end angle=#3,
        radius=#4
      ];

    \pgfmathsetmacro{\middleangle}{(#2+#3)/2}

    \node[
      #6,
      inner sep=0pt
    ] at (\middleangle:{#4+0.35}) {#5};
  \end{scope}
}















\node[
  centre,
  minimum size=7mm
] (ut) at (0,0) {};

\draw[thick] (ut.north) -- (ut.south);

\node[] at ([xshift=-1.75mm]ut.center) {$u$};
\node[] at ([xshift= 1.75mm]ut.center) {$v$};


\outtree{ut}{120}{common}
{$\treh_{n,1}$}
{1.20}{0.7}

\node[common] at (-1.4,1.2) {$\iddots$};

\outtree{ut}{170}{common}
{$\treh_{n,1}$}
{1.20}{0.7}

\branchcount{ut}{170}{120}{0.72}
{$\sqrt{n}$}{common}

\node[common] at (0,2.15) {$\dots$};


\outtree{ut}{60}{common}
{$\treh_{n,K-1}$}
{2}{0.60}

\node[common] at (1.45,1.35) {$\ddots$};

\outtree{ut}{20}{common}
{$\treh_{n,K-1}$}
{2}{0.60}

\branchcount{ut}{60}{20}{0.72}
{$\sqrt{n}$}{common}

\outtree{ut}{-48}{expander}
{$\treh_{n,K}$}
{2.70}{0.57}

\node[expander] at (2.5,-1) {$\iddots$};

\outtree{ut}{-8}{expander}
{$\treh_{n,K}$}
{2.70}{0.57}

\branchcount{ut}{-8}{-48}{0.72}
{$\sqrt n$}{expander}

\outtreepartial{ut}{210}{extra}
  {$\treh_{n,K+1}$}
  {$\treh_{n,K}$}
  {3.5}{1.15}
  {1.45}{0.45}

\node[extra] at (-1.5,-2.25) {$\ddots$};

\outtreepartial{ut}{-80}{extra}
  {$\treh_{n,K+1}$}
  {$\treh_{n,K}$}
  {3.5}{1.15}
  {1.45}{0.45}

\branchcount{ut}{-150}{-80}{0.72}
  {$n-\sqrt n$}{blue!75!black}






\begin{scope}[on background layer]







\fill[red!4]
(ut.center) circle[radius=4cm];

\fill[green!3]
(ut.center)
-- ++(190:4cm)
arc[
start angle=190,
end angle=300,
radius=4cm
]
-- cycle;

\fill[yellow!7]
(ut.center)
-- ++(5:4cm)
arc[
start angle=5,
end angle=-60,
radius=4cm
]
-- cycle;

\draw[gray!45,thin]
(ut.center) circle[radius=4cm];

\draw[gray!35,thin]
($(ut.center)+(190:4cm)$)
--
(ut.center)
--
($(ut.center)+(300:4cm)$)
--
(ut.center)
--
($(ut.center)+(5:4cm)$);

\end{scope}


\begin{scope}[shift={(4.6,1.0)}]

\fill[red!4] (0,0) rectangle (0.35,0.22);
\draw[gray!75,thin] (0,0) rectangle (0.35,0.22);
\node[anchor=west] at (0.45,0.11)
{Decorations--decorations match};

\fill[green!3] (0,-0.7) rectangle (0.35,-0.48);
\draw[gray!75,thin] (0,-0.7) rectangle (0.35,-0.48);
\node[anchor=west] at (0.45,-0.59)
{Expander--expander match};

\fill[yellow!7] (0,-1.4) rectangle (0.35,-1.18);
\draw[gray!75,thin] (0,-1.4) rectangle (0.35,-1.18);
\node[anchor=west] at (0.45,-1.29)
{Expander--decoration match};

\end{scope}
\end{tikzpicture}
    \caption{
      Schematic of the rooted embedding of $B_{\main,\ell_{n,K}}(u)$ into $B_{\envl,\ell_{n,K}}(v)$. The decorations $\treh_{n,k}$, $k=1,\ldots,K-1$, attached to $u$ are mapped canonically to the corresponding decorations at~$v$ (red region). Of the $n$ expander-generated copies of $\treh_{n,K}$ at $u$, $n-\sqrt n$ are mapped into the truncated copies of $\treh_{n,K+1}$ at $v$ using \Cref{lem:embed-ss} (green region); the solid blue inner triangles indicate the embedded copies of $\treh_{n,K}$ inside these larger dotted trees. The remaining $\sqrt n$ expander-generated copies are mapped canonically onto the $\sqrt n$ copies of $\treh_{n,K}$ attached to $v$ as decorations (yellow region). Thus, red, green, and yellow shading indicate decoration--decoration, expander--expander, and expander--decoration matches, respectively. 
    }
\label{Fig:partiso}
  \end{center}
\end{figure}

\begin{corollary}[\sc Embedding of local structures]
  \label{prop:embed-local}
  Let $u \in V_E(\main)$ and $v \in V_E(\envl)$ be expander vertices. There exists an embedding of truncated neighborhoods, $B_{\main,\ell_{n,K}}(u)$ into $B_{\envl,\ell_{n,K}}(v)$ that maps $u$ to $v$ and preserves the degrees of the nodes, except at some of the expander vertices at distance $\ell_{n,K} + 1$ from $u$.
\end{corollary}

\begin{proof}
  By \Cref{Lem:locss} and \Cref{Fig:expanderlocal}, the two graphs differ in the $n-\sqrt{n}$ copies of~$\treh_{n,K+1}$ attached in $B_{\envl,\ell_{n,K}}(v)$ in place of~$\treh_{n,K}$ in $B_{\main,\ell_{n,K}}(u)$. By \Cref{lem:embed-ss}, there exists a rooted embedding of the latter into the former, under which the degree of every vertex is preserved except at some 0-level leaves of the copies of~$\treh_{n,K}$. By \Cref{Lem:locss}, these leaves correspond in $B_{\main,\ell_{n,K}}(u)$ to expander vertices at distance $\ell_{n,K} + 1$ from $u$.
\end{proof}

  \subsection{Proof of hardness of GGSE}
  We use the embedding established in the previous section to show the hardness of distinguishing $\main$ from $\envl$, even when given multiple expander vertices that are far apart.

\begin{proposition}
  \label{Prop:distgraph}
  Fix an integer $t$. For each $G \in \set{\main,\envl}$, let $U_G \subset V_E(G)$ be a set of expander vertices of size $\abs{U_G} = t$ such that $\dist_E(u,v) > 2(\ell_{n,K} + 1)$ for all $u,v \in U_G$, where~$K = \sqrt{n}$.

  Suppose that a classical algorithm is given access to the adjacency-list oracle $\ora_{G,\pi}$, for a random $G \in \set{\main,\envl}$ and a random~$\pi$, together with $\pi(U_G)$, and that it makes $q$ queries. If~$q + t < 2^{\sqrt{n}}$ then the algorithm cannot correctly identify $G$ with probability greater than~$1/2 + 2^{-3n/2}$.
\end{proposition}

\begin{proof}
  By \Cref{lem:conncomp}, except with probability $q \cdot 2^{-n^3\log^2 n}$, we may assume that the vertices queried by the algorithm induce a collection of connected components rooted at the vertices in~$U_G$, together with some isolated vertices. 

  Conditioned on this event, since $\dist_E(u,v) > 2(\ell_{n,K} + 1)$ for all $u,v \in U_G$, the local exploration structures around the vertices in $U_G$ remain disjoint at least until the algorithm queries an expander vertex at distance $\ell_{n,K} + 1$ from some $u \in U_G$. Moreover, by \Cref{prop:embed-local}, the local structure around $u$ is indistinguishable between $G=\main$ and $G=\envl$ at least until such a vertex is queried. By \Cref{Lem:locss}, the probability of this event is upper-bounded by the probability of finding a $0$-level leaf in the tree $\treh_{n,K}$, which is at most $2^{-2n\log{n}}$ by \Cref{Lem:hardss}. By a union bound over all $u \in U_G$, the probability of querying any vertex at distance $\ell_{n,K} + 1$ from $U_G$ is therefore at most $t \cdot 2^{-2n\log{n}}$.

  By standard arguments relating the probability of distinguishing two graphs to the probability that their query responses first differ (see, e.g., \cite[Lemma 5]{Has21j}), we conclude that the probability of correctly identifying $G$ is at most $q \cdot 2^{-n^3\log^2 n} + (1/2 + t \cdot 2^{-2n\log{n}}/2) < 1/2 + 2^{-3n/2}$ since~$q + t < 2^{\sqrt{n}}$.
\end{proof}

We are now ready to prove that the two families of GGSE instances $\main$ and $\envl$ cannot be distinguished efficiently by a classical algorithm.

\begin{proof}[Proof of \Cref{Thm:GEEHard}]
  We first show that the input samples satisfy the spread out condition needed to apply \Cref{Prop:distgraph} with high probability. Let $u_i = \pi^{-1}(x_i)$ for all $i$ be the vertices sampled under the guiding state $\ket{\psi_{E,\pi}}$ (i.e., the uniform superposition over the expander vertices). By the same arguments as in \Cref{Lem:hfaridealpair}, the probability that two vertices are close to each other is, $\pr{\dist_E(u_i,u_j)\leq 2(\ell_{n,K} + 1) \text{ for some }i\neq j} \leq \binom{t}{2} \cdot \frac{n^{2(\ell_{n,K} + 1)+1}}{N_E} \leq 2^{-n^2/2}$ since $t < 2^{\sqrt{n}}$, $\ell_{n,K} = 10n^{2}\log{n}$ and $N_E \geq 2^{100 n^2 \log^2 n}$. Thus we can assume that with probability at least~$1 - 2^{-n^2/2}$, the set $U = \set{u_1,\dots,u_t}$ of input samples satisfies the spread out condition needed to apply \Cref{Prop:distgraph}. Conditioned on this event, the classical algorithm cannot distinguish the two graphs with probability greater than $1/2 + 2^{-3n/2}$. Thus, the overall distinguishing probability is at most $1/2 + 2^{-3n/2} + 2^{-n^2/2} < 1/2 + 2^{-n}$.

  The proof that the guiding state has overlap $\abs{\ip{\psi_{E,\pi}}{\psi_{G,\pi}}} = 1-O(1/n)$ follows directly from \Cref{Lem:spec}.
\end{proof}

\section*{Acknowledgements}
We acknowledge the use of AI-based tools for language editing and proofreading, and for auditing the proofs after the first version of this paper was posted on arXiv. This audit helped us identify and correct minor errors.

This work was supported by the Maison du Quantique de Nouvelle-Aquitaine ``HybQuant'', as part of the HQI initiative and France 2030, under the French National Research Agency (ANR) grant ANR-22-PNCQ-0002. It was also supported by the PEPR integrated project EPiQ under ANR grant ANR-22-PETQ-0007.

\printbibliography[heading=bibintoc]

\appendix

\section{Construction of high-degree, high-girth expanders}
\label{App:app-graph}

  \subsection{Construction via LPS Ramanujan graphs}
  Lubotzky, Phillips and Sarnak~\cite{LPS88j} constructed a family of high-degree, high-girth Ramanujan Cayley graphs $(X^{p,q})_{p,q}$, parametrized by pairs of prime numbers $p$ and $q$ congruent to $1$ modulo $4$. We show that these prime numbers can be chosen so that deleting suitable pairs of generators yields the graphs required by \Cref{Prop:expander}. First, we recall the result from~\cite{LPS88j}.

\begin{theorem}[\cite{LPS88j}]
  \label{Thm:LPS}
  Let $p < q$ be two prime numbers congruent to $1$ modulo~$4$. Then, the Ramanujan graph $X^{p,q}$ constructed in \cite{LPS88j} has the following properties:
  \begin{itemize}
    \item The number of vertices is either $q(q^2-1)/2$ or $q(q^2-1)$,
    \item The graph is regular of degree $p+1$,
    \item The spectral gap between its two largest adjacency eigenvalues is at least $(p+1)-2\sqrt{p}$,
    \item The girth is at least $2\log_p q$.
  \end{itemize}
\end{theorem}

We now prove that \Cref{Prop:expander} holds for every sufficiently large even $n$.

\begin{proposition}
  For every sufficiently large even integer $n$, there exists a graph~$E_n$ on $N_E \in [2^{100n^{5/2}\log^2 n},2^{n^3\log^2 n}]$ vertices with degree~$d_n=n$, spectral gap at least $n/2$ and girth at least $g_n=100n^{5/2}\log n$.
\end{proposition}

\begin{proof}
  The proof makes use of the following variant of Bertrand's postulate for primes of prescribed congruence classes.

  \begin{lemma}[{\cite[Theorem 1]{CH12j}}]
    \label{Lem:cong}
    For any $x\geq 887$, the interval $(x,1.048x]$ contains a prime congruent to~$1\bmod4$.
  \end{lemma}

  Let $n$ be a sufficiently large even integer. By \Cref{Lem:cong}, choose a prime $p\equiv1\pmod4$ in $(n,1.048n]$. Set $B=60n^{5/2}\log^2 n$ and choose a prime $q\equiv1\pmod4$ in $(2^B,1.048\cdot2^B]\subset(2^B,2^{B+1}]$ using the same lemma.

  The generating set of the Cayley graph $X^{p,q}$ splits into $(p+1)/2$ pairs $\{s,s^{-1}\}$, each consisting of a generator and its inverse~\cite[Section~3]{LPS88j}. Removing such a pair preserves the undirectedness of the graph and decreases its degree by two. Retain any $n/2$ of these pairs and let~$E_n$ be the corresponding Cayley graph on the same vertex set. The deleted edges form a $d'$-regular graph $Y$, where $d'=p+1-n \leq 0.048n+1$, so $A_{E_n}=A_{X^{p,q}}-A_Y$ and $\norm{A_Y}\leq d'$. Applying the variational principle on the subspace orthogonal to the uniform vector, its second-largest adjacency eigenvalue satisfies $\lambda_2(E_n)\leq2\sqrt p+d' \leq2\sqrt{1.048n}+0.048n+1 \leq n/2$ for sufficiently large $n$. Thus the spectral gap of $E_n$ is at least $n/2$.

  By \Cref{Thm:LPS}, the number of vertices satisfies $N_E\in\set{q(q^2-1)/2,q(q^2-1)}$, and hence $\log N_E=3B+O(1)=180n^{5/2}\log^2 n+O(1)$. Finally, deleting edges cannot decrease the girth, so it is at least $2\log_p q \geq 120n^{5/2}\log^2(n) / \log(1.048n) \geq 100n^{5/2}\log n$.
\end{proof}

  \subsection{Construction via short cycle removal}
  We use the following results of Paredes~\cite{Par21c}. We first recall the configuration model and the relevant definitions. Throughout, loops contribute two to the degree of a vertex.

\begin{definition}[Configuration model, {\cite[Definition~14]{Par21c}}]
  \label{def:configuration-model}
  Let $N>d>0$ be integers with $Nd$ even. The $d$-regular $N$-vertex configuration model is obtained by taking the set of half-edges $[N]\times[d]$ and choosing a uniformly random perfect matching of these $Nd$ half-edges. Each matched pair is then joined to form an edge between the corresponding vertices. The resulting object is a $d$-regular undirected multigraph, where loops and parallel edges are allowed.
\end{definition}

For a $d$-regular multigraph $G$, let $d=\lambda_1(G)\ge\lambda_2(G)\ge\cdots\ge\lambda_N(G)$ denote the eigenvalues of its adjacency matrix, and write $\lambda(G):=\max\{\lambda_2(G),|\lambda_N(G)|\}$.

\begin{definition}[{\cite[Definitions~6 and~7]{Par21c}}]
  \label{def:bicycle-free}
  A multigraph $G$ is \emph{bicycle-free at radius $r$} if the radius-$r$ neighborhood of every vertex contains at most one cycle.

  For positive integers $r$ and $\tau$, we call a (multi)graph $G$ an $(r,\tau)$-(multi)graph if it is bicycle-free at radius at least $r$ and contains at most $\tau$ cycles of length at most $r$. We denote the collection of such cycles by $\cyfc_r(G)$.
\end{definition}

The following result gives bicycle-freeness and spectral expansion simultaneously, while allowing the degree to grow with the number of vertices.

\begin{theorem}[{\cite[Theorem~32]{Par21c}}]
  \label{thm:paredes-random}
  There exists a sufficiently large universal constant $\alpha$ such that the following holds. Let $N>0$, and suppose that $3\le d\le\alpha^{-1}\sqrt{\log N}$,
  $c<1/4$, and \[
  \alpha^3
  \left(
  \frac{\log\log N}{\log_{d-1}N}
  \right)^2
  \le \epsilon\le1.
  \]
  If $G$ is sampled from the $d$-regular $N$-vertex uniform configuration model, then, except with probability at most $N^{-0.99}$, $G$ is bicycle-free at radius $c\log_{d-1}N$ and
  \[
  \lambda(G)\le2\sqrt{d-1}(1+\epsilon).
  \]
\end{theorem}

We also use the following bound on the number of short cycles. We use the multigraph versions of the following results, as justified by the configuration-model extension in \cite[Remark~29]{Par21c}.

\begin{proposition}[{\cite[Proposition~34]{Par21c}}]
  \label{prop:paredes-cycles}
  Let $G$ be a $d$-regular $N$-vertex multigraph that is bicycle-free at radius $2r$. Then
  \[
  |\cyfc_r(G)|\le\frac{N}{(d-1)^r}.
  \]
\end{proposition}

Finally, we use the short-cycle removal construction of
\cite[Theorem~8]{Par21c}. Since in our application $d$ grows with $n$, we state the bounds obtained by keeping explicit the dependence on $d$ in its construction and proof.

\begin{theorem}[Short-cycle removal, adapted from {\cite[Theorem~8 and its proof]{Par21c}}]
  \label{thm:paredes-fix}
  Let $G$ be a $d$-regular $N$-vertex $(r,\tau)$-multigraph satisfying $r\le(2/3)\log_{d-1}(N/\tau)-5$. Set $h:=
  \left\lceil\log_{d-1}\tau\right\rceil
  +
  \left\lceil r/2\right\rceil+1$ and
  \[
  T(d,h):=
  1+d\sum_{j=0}^{h-1}(d-1)^j
  =
  1+\frac{d((d-1)^h-1)}{d-2}.
  \]
  Then the deterministic algorithm $\fix$ of \cite{Par21c} runs in time $\poly(N)$ and outputs a $d$-regular graph $\fix(G)$ of girth at least $r$ such that
  \[
  |V(\fix(G))|\le N+2T(d,h)
  \]
  and
  \[
  \lambda(\fix(G))
  \le
  \max\{\lambda(G),2\sqrt{d-1}\}
  +
  \frac{2(\sqrt{d-1}+1)}{\lfloor r/2\rfloor}
  +
  \frac{2dT(d,h)}{N}.
  \]
\end{theorem}

We can now prove \Cref{Prop:expander}, which we restate below.

\begin{proposition}
  \label{Prop:expander-paredes}
  There exists a family of $n$-regular graphs $(E_n)_n$ on
  $N_E=2^{n^3}(1+o(1))$ vertices such that, for all sufficiently large $n$, $E_n$ has spectral gap at least $n/2$ and girth at least $g_n=100n^{5/2}\log n$.
\end{proposition}

\begin{proof}
Let $N:=2^{n^3}$, $d:=n$, and $L:=\log_{d-1}N$. Let $\alpha$ be the universal constant from Theorem~\ref{thm:paredes-random}. We apply that theorem with $c=1/5$ and $\epsilon=1$.

We first verify its assumptions. Since $N=2^{n^3}$,
$\sqrt{\log N}=\Theta(n^{3/2})$, and therefore
\[
\frac{d}{\sqrt{\log N}}=O(n^{-1/2})\longrightarrow0.
\]
Since $\alpha$ is a fixed constant, it follows that
$d\le\alpha^{-1}\sqrt{\log N}$ for all sufficiently large $n$. Moreover,
\[
\log_{d-1}N
=
\Theta\left(\frac{n^3}{\log n}\right),
\qquad
\log\log N=\Theta(\log n),
\]
and hence
\[
\alpha^3
\left(
\frac{\log\log N}{\log_{d-1}N}
\right)^2
=
O\left(
\alpha^3\frac{\log^4 n}{n^6}
\right)
=o(1).
\]
Thus the latter quantity is at most $\epsilon=1$ for all sufficiently large $n$. The remaining condition $c<1/4$ is immediate.

Theorem~\ref{thm:paredes-random} therefore implies that, with positive probability, a realization $G$ of the $d$-regular $N$-vertex uniform configuration model is bicycle-free at radius $L/5$ and satisfies
\[
\lambda(G)\le4\sqrt{d-1}.
\]
In particular, there exists such a realization $G$.

Set $s:=\lfloor L/10\rfloor$. Since $2s\le L/5$,
Proposition~\ref{prop:paredes-cycles} gives
\[
|\cyfc_s(G)|
\le
\frac{N}{(d-1)^s}.
\]
Define $\tau:= \left\lceil{N}/{(d-1)^s}\right\rceil$ and $r:=\left\lfloor{L}/{20}\right\rfloor$. Since $r\le s$, the graph $G$ is an $(r,\tau)$-multigraph.

We next verify the hypothesis of Theorem~\ref{thm:paredes-fix}. By definition, we have $\tau
\le {2N}/{(d-1)^s}$ and therefore $\log_{d-1}(N/\tau) \ge s-\log_{d-1}2$. As $s=(L/10)+O(1)$, $(2/3)\log_{d-1}(N/\tau)-5 \ge (L/15)-O(1)$. On the other hand, $r\le L/20$. Since $L$ grows with $N$, for all sufficiently large $n$ we have
\[
r
\le
\frac23\log_{d-1}(N/\tau)-5.
\]
Thus Theorem~\ref{thm:paredes-fix} applies. Set $E_n:=\fix(G)$. By Theorem~\ref{thm:paredes-fix}, $E_n$ is $d$-regular and $\girth(E_n)\ge r$. Moreover, $r = \Theta\left({n^3}/{\log n}\right)>100n^{5/2}\log n$ for all sufficiently large $n$. Hence $\girth(E_n)\ge g_n$. However it should be noted that this algorithm runs in $\poly(N) = 2^{\poly(n)}$ time and is thus exponential in $n$, the degree of the graph.

We next bound the number of vertices. Recall from
Theorem~\ref{thm:paredes-fix} that the two $d$-regular trees used in the construction have height
\[
h=
\left\lceil\log_{d-1}\tau\right\rceil
+
\left\lceil r/2\right\rceil+1.
\]
Since $\tau\le2N/(d-1)^s$,
\[
\log_{d-1}\tau
\le
L-s+\log_{d-1}2.
\]
Using $s=L/10+O(1)$ and $r=L/20+O(1)$, we obtain
\[
h
\le
L-s+\frac r2+O(1)
=
\frac{37}{40}L+O(1).
\]
A $d$-regular tree of height $h$ contains
\[
T(d,h)
=
1+\frac{d((d-1)^h-1)}{d-2}
\le
1+3(d-1)^h
\]
vertices for $d\ge3$. Consequently,
\[
T(d,h)
\le
(d-1)^{\frac{37}{40}L+O(1)}
=
N^{37/40+o(1)},
\]
where we used $d=n=N^{o(1)}$. The explicit size bound in
Theorem~\ref{thm:paredes-fix} therefore gives
\[
N
\le
|V(E_n)|
\le
N+N^{37/40+o(1)}
=
N(1+o(1)).
\]
Thus $N_E=|V(E_n)|=2^{n^3}(1+o(1))$.

Finally, we bound the spectrum. By the explicit form of
Theorem~\ref{thm:paredes-fix},
\[
\lambda(E_n)
\le
\max\{\lambda(G),2\sqrt{d-1}\}
+
\frac{2(\sqrt{d-1}+1)}{\lfloor r/2\rfloor}
+
\frac{2dT(d,h)}{N}.
\]
Since $\lambda(G)\le4\sqrt{d-1}$, the first term is at most
$4\sqrt{d-1}$. Furthermore,
\[
\frac{\sqrt d}{r}
=
O\left(\frac{\log n}{n^{5/2}}\right)
=o(1),
\]
while
\[
\frac{dT(d,h)}{N}
\le
n\,N^{-3/40+o(1)}
=
N^{-3/40+o(1)}
=o(1).
\]
It follows that $\lambda(E_n) \le 4\sqrt{d-1}+o(1)$. In particular, since $\lambda_2(E_n)\le\lambda(E_n)$,
\[
d-\lambda_2(E_n)
\ge
d-4\sqrt{d-1}-o(1)
\ge
\frac d2
\]
for all sufficiently large $n$. Since $d=n$, the spectral gap of
$E_n$ is at least $n/2$, as claimed.
\end{proof}

The same result holds even when $d = n-\sqrt{n}$ with the same girth by a similar argument. This lets us prove \Cref{Prop:highgirth}.

\section{Spectral properties of the decorated graph}
\label{App:app-spectral}

In this section, we prove the spectral properties of the graph $G_{n,E}$ that are summarized in \Cref{Lem:spec}. The study of the top eigenvector of its adjacency matrix is closely related to the analysis in Appendix~B of \cite{GHV21c}. However, their argument proceeds via an iterative decoration process in which identical complete trees are attached to all vertices at each stage. Our construction differs in that we do not decorate the leaves created at intermediate levels of the iteration as in Definitions~\ref{Def:selfsim} and \ref{Def:maingraph}. Consequently, the spectral analysis in \cite{GHV21c} does not apply directly to our setting. We therefore present a self-contained proof that analyzes the top-eigenvector of the main graph directly, rather than via an iterative argument.

Below is a lemma proved in \cite{Has21j} about the spectral gap of the sum of two Hermitian matrices of the same dimension in terms of the spectral gap of one and the spectral norm of another.

\begin{lemma}[\sc Spectral gap of sum of matrices, (Lemma 4 in \cite{Has21j}; Footnote 15 in \cite{GHV21c})]
  \label{Lem:app-specsumhastings}
  Let $A_1$ and $A_2$ be any two Hermitian matrices of the same dimension. If $A_1$ has a spectral gap of $\delta$ and $A_2$ has spectral norm at most $\gamma$, then $A_1 + A_2$ has a spectral gap of at least $\delta - 2\gamma$.
\end{lemma}

Applying the above result to the main graph, we obtain the following, which proves the first part of \Cref{Lem:spec}:

\begin{lemma}[\sc Spectral gap of $G_{n,E_n}$]
  \label{lem:app-spec}
  The difference between the two largest eigenvalues of $A_{n,E_n}$ is at least $d_n/2 - 4\sqrt{2n}$.
\end{lemma}

\begin{proof}
  The adjacency matrix $A_{n,E_n}$ of the graph $G_{n,E_n}$ can be written as $A_{n,E_n} = A_{E} + A_{T}$, with~$A_{E}$ denoting the expander part and $A_{T}$ denoting the tree edges added to it. By the assumptions of \Cref{Prop:expander}, the spectral gap of $A_{E}$ is at least $n/2$. The graph formed by the added tree edges is a forest of maximum degree $2n$, and thus the corresponding adjacency matrix $A_T$ has spectral norm bounded by $2\sqrt{2n}$.

  Applying \Cref{Lem:app-specsumhastings} with $A_1 = A_{E}$, $A_2 = A_{T}$, $\delta = n/2$ and $\gamma = 2\sqrt{2n}$ proves the lemma.
\end{proof}

We next develop a general description of the Perron--Frobenius eigenpair of a graph obtained by attaching the same collection of rooted trees to every vertex of a connected base graph.

Let $E$ be a connected graph whose adjacency matrix is denoted by $A_E$ and $\mcl{T}$ be a multiset of finite rooted trees where the adjacency matrix of $T \in \mcl{T}$ is denoted by $A_T$. Consider the graph $G$ obtained by attaching every $T \in \mcl{T}$ to $v \in V_E$ through an edge $\{v,t\}$ where $t$ is the root of $T$. Let $A_G$ be its adjacency matrix. 

The results below express the Perron--Frobenius eigenvalue and eigenvector of $G$ in terms of those of $E$ and the attached trees. We will later apply them to $G_{n,E}$.

\begin{lemma}[\sc Top eigenvector under rooted-tree decorations]
\label{Lem:app-generic-dectop}
    Fix a Perron--Frobenius vector $\psi_E$ with arbitrary normalization,
    and let $\lambda_E$ denote the corresponding eigenvalue of the
    adjacency matrix $A_E$. Let $\mathcal{T}$ be a finite multiset of
    finite rooted trees. For $T\in\mathcal{T}$, denote its root by $t$.
    Denote the Perron--Frobenius eigenvalue of $A_G$, the adjacency matrix
    of $G$, by $\lambda_G$.

    Then we have the following:
    \begin{enumerate}
        \item $\lambda_G
        =
        \lambda_E
        +
        \sum_{T\in\mathcal T}
        \big[(\lambda_G\cdot I-A_T)^{-1}\big]_{t,t}$.

        \item The restriction of any Perron--Frobenius vector of $G$ to
        $E$ is proportional to $\psi_E$. Let $\psi_G$ be the
        Perron--Frobenius vector of $G$ whose restriction to $E$ exactly
        equals $\psi_E$. Then,
        \[
        \psi_G
        =
        \psi_E
        \;\oplus\;
        \bigoplus_{v\in V_E}
        \psi_E(v)
        \left(
        \bigoplus_{T\in\mathcal T}
        (\lambda_G\cdot I-A_T)^{-1}e_t
        \right).
        \]
    \end{enumerate}
\end{lemma}

\begin{proof}
    For each $T\in\mathcal T$, we use the same notation for every copy of $T$ attached to a vertex of $E$. In particular, we denote the root of each copy by $t$, and its adjacency matrix by $A_T$.

    The adjacency matrix of $G$ has the block form
    \[
    A_G=
    \begin{pmatrix}
    A_E & B\\
    B^{\top} & \displaystyle\bigoplus_{v\in V_E}
    \bigoplus_{T\in\mathcal T}A_T
    \end{pmatrix},
    \]
    where $B$ records the edges joining each vertex $v\in V_E$ to the
    root $t$ of every copy of $T$ attached to it.

    Let $\psi_G=(x,y)$ be a Perron--Frobenius eigenvector of $A_G$, where
    $x\in\mathbb R^{V_E}$ and
    $y=\bigoplus_{v\in V_E}\bigoplus_{T\in\mathcal T}y_{v,T}$.
    The eigenvalue equation $A_G\psi_G=\lambda_G\cdot\psi_G$ is equivalent to
    \[
    A_Ex+By=\lambda_Gx
    \]
    and
    \[
    B^{\top}x+
    \left(\bigoplus_{v,T}A_T\right)y
    =
    \lambda_Gy.
    \]

    We first solve the second equation. Since the blocks are independent,
    for every pair $(v,T)$ we obtain $(\lambda_G\cdot I-A_T)y_{v,T}=e_tx_v$.
    The matrix $\lambda_G\cdot I-A_T$ is invertible because
    $\lambda_G>\lambda_T$.

    To see this, let $z$ be a Perron--Frobenius eigenvector of $A_T$, and
    extend it by zero outside $T$ to obtain a vector $\widetilde z$ on
    $G$. The Rayleigh quotient of $\widetilde z$ is $\lambda_T$, so
    $\lambda_T\leq\lambda_G$. If equality held, then $\widetilde z$ would
    be a top eigenvector of $A_G$. However, if $v$ is the vertex of $E$
    attached to the root of this copy of $T$, then
    $\widetilde z_v=0$, whereas
    $(A_G\widetilde z)_v=\widetilde z_t>0$, contradicting
    $A_G\widetilde z=\lambda_G\widetilde z$. Hence
    $\lambda_T<\lambda_G$.

    Therefore, $y_{v,T} = (\lambda_G\cdot I-A_T)^{-1}e_t\,x_v$. Substituting this into the first block equation gives, for every $v\in V_E$,
    \[
    (A_Ex)_v
    +
    \sum_{T\in\mathcal T}
    \big[(\lambda_G\cdot I-A_T)^{-1}\big]_{t,t}x_v
    =
    \lambda_Gx_v.
    \]
    Equivalently,
    \[
    \left(
    A_E+
    \sum_{T\in\mathcal T}
    \big[(\lambda_G\cdot I-A_T)^{-1}\big]_{t,t}I
    \right)x
    =
    \lambda_Gx.
    \]

    The vector $x$ is positive, so by the Perron--Frobenius theorem it is
    proportional to $\psi_E$. Comparing the corresponding eigenvalues
    gives
    \[
    \lambda_G
    =
    \lambda_E+
    \sum_{T\in\mathcal T}
    \big[(\lambda_G\cdot I-A_T)^{-1}\big]_{t,t},
    \]
    which proves part~(1).

    Finally, choose the normalization $x=\psi_E$. Then
    \[
    y_{v,T}
    =
    \psi_E(v)(\lambda_G\cdot I-A_T)^{-1}e_t.
    \]
    Combining the expressions for $x$ and $y$ gives
    \[
    \psi_G
    =
    \psi_E
    \;\oplus\;
    \bigoplus_{v\in V_E}
    \psi_E(v)
    \left(
    \bigoplus_{T\in\mathcal T}
    (\lambda_G\cdot I-A_T)^{-1}e_t
    \right),
    \]
    which proves part~(2).
\end{proof}

Next we show that most of the Perron mass is supported on the base graph whenever the number and spectral norms of the decorations are sufficiently small relative to $\lambda_G$.

\begin{corollary}[\sc Overlap with the base graph]
\label{Cor:app-base-overlap}
Consider the setting of \Cref{Lem:app-generic-dectop}. Let $\lambda_G$ be the Perron eigenvalue of $G$, let $\psi_G$ be a corresponding Perron--Frobenius eigenvector, and write $\psi_E:=\psi_G|_{V_E}$. Define $\eta:=\max_{T\in\mcl T}\frac{\|A_T\|}{\lambda_G}$. Then,
\[
\frac{\|\psi_E\|_2^2}{\|\psi_G\|_2^2}
\geq
\left(
1+
\frac{|\mcl T|}{\lambda_G^2(1-\eta)^2}
\right)^{-1}.
\]
\end{corollary}

\begin{proof}
  For every $T\in\mcl T$, the strict inequality $\|A_T\|<\lambda_G$ follows from the proof of \Cref{Lem:app-generic-dectop}. Hence $\eta<1$. For $v\in V_E$ and $T\in\mcl T$, let $T_v$ denote the copy of $T$ attached to $v$ with root $t$, and let $\psi_{T_v}$ denote the restriction of $\psi_G$ to $T_v$. By \Cref{Lem:app-generic-dectop},
  \[
  \psi_{T_v}
  =
  \psi_E(v)(\lambda_G\cdot I-A_T)^{-1}e_t.
  \]
  Since $\|A_T\|<\lambda_G$, we have that $(\lambda_G\cdot I-A_T)$ is positive definite. Therefore,
  \begin{align*}
  \left\|
  {\psi_{T_v}}
  \right\|_2
  \leq |\psi_E(v)|\left\|(\lambda_G\cdot I-A_T)^{-1}\right\|_2 \leq \frac{|\psi_E(v)|}{(\lambda_G-\norm{A_T})}\leq 
  \frac{|\psi_E(v)|}{\lambda_G(1-\eta)}.
  \end{align*}
  Since every decoration vertex belongs to a unique copy
  attached to a unique base vertex,
  \[
  \|\psi_G\|_2^2-\|\psi_E\|_2^2 \leq \sum_{v\in V_E}\sum_{T_v \in \mcl{T}}
  \frac{\psi_E(v)^2}{\lambda_G^2(1-\eta)^2} \leq
  \frac{|\mcl T|}{\lambda_G^2(1-\eta)^2}
  \|\psi_E\|_2^2.
  \]
  Rearranging proves the claim.
\end{proof}

We now specialize the generic estimates to the two graph families used in the construction. We begin with the main graph $G_n$, for which the base graph is $n$-regular and the decorations are the trees $\treh_{n,k}$ with $k\in[\sqrt{n}-1]$.

\begin{corollary}[\sc Top-eigenvalue and expander overlap for $G_n$ and $F_n$]
\label{Cor:app-main-concentration}
    Let $\lambda_{n,E}$ be the Perron eigenvalue of $G_{n,E}$ for $(K,E) \in \set{(\sqrt{n},E_{n}),(\sqrt{n}+1,E'_{n})}$ (cf.\ \Cref{Def:maingraph}), and let $\psi_{n,E}$ be the corresponding Perron--Frobenius eigenvector. Write $\psi_{E}:=\psi_{n,E}|_{V(E)}$. Then,
    \begin{enumerate}
        \item $\lambda_{n,E}=d_{n,K}+O(1)$.
        \item $\|\psi_{E}\|_2^2\geq(1-O(1/n))\cdot\|\psi_{n,E}\|_2^2$.
    \end{enumerate}
\end{corollary}

\begin{proof}
  At each vertex of $E$, the multiset of attached trees is
  $\mcl T = \biguplus_{k=1}^{K-1}\biguplus_{j=1}^{\sqrt{n}}\{\treh_{n,k}\}$, and hence $|\mcl T|=\sqrt{n}\cdot(K-1)$. By \Cref{Lem:app-generic-dectop}, $\lambda_{n,E}=d_{n,K}+\sum_{T\in\mcl T}\big[(\lambda_{n,E} \cdot I-A_T)^{-1}\big]_{t,t}$.
  Every tree in $\mcl T$ has maximum degree at most $2n$, so $\|A_T\|\leq2\sqrt{2n}$. Since $\lambda_{n,E}\geq d_{n,K}$, the matrix $\lambda_{n,E} \cdot I-A_T$ is positive definite, and
  \[
    0\leq\big[(\lambda_{n,E} \cdot I-A_T)^{-1}\big]_{t,t}
    \leq\|(\lambda_{n,E} \cdot I-A_T)^{-1}\|
    \leq\frac{1}{d_{n,K}-2\sqrt{2n}}.
  \]
  Here the first inequality follows by writing the diagonal entry as $e_t^\top(\lambda_{n,E} \cdot I-A_T)^{-1}e_t$, while the last follows since the
  eigenvalues of $(\lambda_{n,E} \cdot I-A_T)^{-1}$ are $(\lambda_{n,E}-\mu)^{-1}$ for $\mu$ in the spectrum of $A_T$. Therefore $0\leq\lambda_{n,E}-d_{n,K}\leq\sqrt{n}\cdot(K-1)/(d_{n,K}-2\sqrt{2n})=O(1)$. When $E \in \{E_n,E'_n\}$, we get that $\lambda_{n,E}=d_{n,K}+O(1)$ and $\max_{T\in\mcl T}\|A_T\|/\lambda_{n,E} \leq 2\sqrt{2n}/d_{n,K}=O(1/\sqrt{n})$. The overlap bound follows from \Cref{Cor:app-base-overlap}.
\end{proof}

\section{Hardness of exploring self-similar trees}
\label{App:trees}
This section is devoted to proving that any classical algorithm requires a sub-exponential number of queries to explore a $0$-level leaf in the full-level decorated tree $\treh_{n,k}$ starting from its root (see \Cref{Def:selfsim} for a description of these objects). The main result is \Cref{Lem:hardss} (restated below as \Cref{Prop:exithard}). The proof proceeds by induction on the index~$k$. 

We begin by introducing several probabilistic events and random variables that play a central role in the argument. Observe that $\exih_{n,K}$ for $K = \sqrt{n}$ is precisely the event whose probability we aim to bound in order to prove \Cref{Lem:hardss}.

\begin{definition}[\sc Events and random variables]
  Given a randomized algorithm that explores the $1$-level decorated tree $T_{n,k,1}$ starting at its root, we define:
    \begin{itemize}
      \item $\exi_{n,k}$: the event that the algorithm queries a $0$-level leaf,
      \item $\lea_{n,k}$: the random variable denoting the number of queried $1$-level leaves that belong to \emph{distinct} decorations.
    \end{itemize}
  The same objects are denoted $\exih_{n,k}$ and $\leah_{n,k}$ when the algorithm is instead exploring the full-level decorated tree $\treh_{n,k}$.
\end{definition}

For a randomized algorithm $\alg$, let $q_{\alg}$ denote the maximum number of queries made by that algorithm (over all possible random executions). Our goal is to identify the largest value of~$q_{n,k}$ for which the event $\exih_{n,k}$ is unlikely to occur when the algorithm makes at most $q_{n,k}$ queries. This amounts to upper-bounding the quantity $\max_{\alg : q_{\alg} < q_{n,k}} \pr{\exih_{n,k}}$. The inductive proof will relate this to $\max_{\alg : q_{\alg} < q_{n,k-1}} \pr{\exih_{n,k-1}}$, where the number of queries is decreased to $q_{n,k-1} < q_{n,k}$ and the input tree is $\treh_{n,k-1}$.

First, it is easy to see that full-level decorated trees are harder to explore than 1-level decorated trees. In particular, we will make use of the following lemma.

\begin{lemma}[\sc Monotonicity of hardness under decorations]
  \label{Lem:dec}
  For any number $q$ of queries and integer $w$, we have
  \[\max_{\alg : q_{\alg} < q} \pr{\exih_{n,k}\ \mathrm{and}\ \leah_{n,k} < w} \leq \max_{\alg : q_{\alg} < q} \pr{\exi_{n,k}\ \mathrm{and}\ \lea_{n,k} < w}.\]
\end{lemma}

\begin{proof}
  Any algorithm with query access to $T_{n,k,1}$ can simulate query access to the full-level decorated tree $\treh_{n,k}$ by simply simulating the subsequent levels of decoration itself. Hence, the success probability cannot be larger on $\treh_{n,k}$ than on $T_{n,k,1}$.
\end{proof}

The next central lemma upper bounds the probability of the event $\exi_{n,k}$ occurring when only a few $1$-level leaves of the 1-level decorated tree $T_{n,k,1}$ are explored. Note that the input graph need not be fully-decorated, nor does the number $q_{\alg}$ of queries need to be bounded to apply this lemma.

\begin{lemma}[\sc Hardness of avoiding the decorations]
  \label{Lem:strght}
  For any integer $w$,
  \[\max_{\alg : q_{\alg} < \infty} \pr{\exi_{n,k}\ \mathrm{and}\ \lea_{n,k} < w} \leq \pt*{\frac{d_{n,k}}{d_{n,k-1}}}^{(\ell_{n,k} - \ell_{n,k-1})/w}.\]
\end{lemma}

\begin{proof}
  The cases $w = 1$ and $w = 2$ correspond essentially to \cite[Lemma 3]{GHV21c} and \cite[Lemma 4, Case 1]{GHV21c}.
  We do not prove the case $w > 2$ here since it is not needed for our application.

  We sketch the proof for $w = 1$. We view an algorithm $\alg$ as a local exploration that starts at the root of $T_{n,k,1}$ and may only query edges adjacent to previously queried vertices. The events $\exi_{n,k}\ \mathrm{and}\ \lea_{n,k} < 1$ require the first queried leaf in $T_{n,k,1}$ to be a $0$-level leaf. Thus, the algorithm must grow an exploration path of length $\ell_{n,k}$ that reaches a $0$-level leaf before any $1$-level leaf is queried. A path that grows to length $\ell_{n,k}$ avoids reaching a $1$-level leaf if and only if it does not enter any decoration before depth $\ell_{n,k} - \ell_{n,k-1}$, since the decorations have depth~$\ell_{n,k-1}$. For each node in the base graph, the proportion of its children that do not belong to a decoration is at most $d_{n,k}/d_{n,k-1}$. Hence, the probability of avoiding any decoration for $\ell_{n,k} - \ell_{n,k-1}$ consecutive steps is at most $(d_{n,k}/d_{n,k-1})^{\ell_{n,k} - \ell_{n,k-1}}$.
\end{proof}

The next lemma exploits the self-similar structure of $\ha{T}_{n,k}$ to relate the probability of finding multiple $1$-level leaves in $\ha{T}_{n,k}$ to the probability of the event $\exih_{n,k-1}$. Most importantly, the number of queries allowed to explore~$\ha{T}_{n,k-1}$ is reduced by a factor of $1/w$, which will make the recursion effective later on.

\begin{lemma}[\sc Self similarity]
  \label{Lem:selfs}
  For any integer $w$,
  \[\max_{\alg : q_{\alg} < q} \pr{\leah_{n,k} \geq w} \leq q \cdot \max_{\alg : q_{\alg} < q/w} \pr{\exih_{n,k-1}}.\]
\end{lemma}

\begin{proof}
  This generalizes \cite[Lemma 4, Case 2]{GHV21c}. Essentially, finding leaves of $w$ distinct outer trees in $\treh_{n,k}$ amounts to the event $\exih_{n,k-1}$ occurring $w$ times in $w$ distinct subtrees $\treh_{n,k-1}$ of~$\treh_{n,k}$. By a pigeonhole argument, at least one of these events must happen after making only $q/w$ queries in the corresponding subtree $\treh_{n,k-1}$. Since at most $q$ subtrees are explored in total (this is a very crude bound), by a union bound argument, the overall probability must be at most $q \cdot \max_{\alg : q_{\alg} < q/w} \pr{\exih_{n,k-1}}$.
\end{proof}

Combining the above three results, we obtain the inductive step of the proof:

\begin{corollary}[\sc Inductive argument]
  \label{Prop:ind}
  For any integer $w$,
  \begin{equation}
    \label{Eq:ind}
    \max_{\alg : q_{\alg} < q_{n,k}} \pr{{\exih_{n,k}}} \leq \pt*{\frac{d_{n,k}}{d_{n,k-1}}}^{(\ell_{n,k} - \ell_{n,k-1})/w} + q_{n,k} \cdot \max_{\alg : q_{\alg} < q_{n,k}/w} \pr{\exih_{n,k-1}}.
  \end{equation}
\end{corollary}

\begin{proof}
  It is simply a matter of applying \Cref{Lem:dec,Lem:selfs,Lem:strght}, and basic laws of probability,
  \begin{align*}
    \max_{\alg : q_{\alg} < q_{n,k}} \pr{\exih_{n,k}}
     & \leq \max_{\alg : q_{\alg} < q_{n,k}} \pr{\exih_{n,k}\ \mathrm{and}\ \leah_{n,k} < w} + \max_{\alg : q_{\alg} < q_{n,k}} \pr{\exih_{n,k}\ \mathrm{and}\ \leah_{n,k} \geq w} \\
     & \leq \max_{\alg : q_{\alg} < \infty} \pr{\exi_{n,k}\ \mathrm{and}\ \lea_{n,k} < w} + \max_{\alg : q_{\alg} < q_{n,k}} \pr{\leah_{n,k} \geq w} \\
     & \leq \pt*{\frac{d_{n,k}}{d_{n,k-1}}}^{(\ell_{n,k} - \ell_{n,k-1})/w} + q_{n,k} \cdot \max_{\alg : q_{\alg} < q_{n,k}/w} \pr{\exih_{n,k-1}}.
  \end{align*}
\end{proof}

Finally, we solve the above recurrence relation using $w = 2$ and the set of parameters given in \Cref{Lem:hardss}.

\begin{proposition}[\sc Finding exit is hard - \Cref{Lem:hardss} restated]
  \label{Prop:exithard}
  For $k \in \set{1,\dots,\sqrt{n}}$, let $d_{n,k} = 2n - k \sqrt{n}$, $\ell_{n,k} = k\cdot10n^{3/2}\log{n}$, $q_{n,k} = 2^k$ and $w = 2$. Then,
    \[\max_{\alg : q_{\alg} < q_{n,k}} \pr{\exih_{n,k}} \leq 2^{-2n\log{n}+(\sqrt{n}+1)(k-\sqrt{n})}.\]
  In particular, for $K = \sqrt{n}$, the probability of $\exih_{n,K}$ is at most $2^{-2n\log{n}}$ for any classical algorithm making fewer than $2^{\sqrt{n}}$ queries. Moreover, the tree $\treh_{n,K}$ contains fewer than $2^{12n^{3}\log^2{n}}$ vertices.
\end{proposition}

\begin{proof}
  The base case $k = 1$ is trivial. For the induction, we have,
  \begin{align*}
    \pt*{\frac{d_{n,k}}{d_{n,k-1}}}^{(\ell_{n,k} - \ell_{n,k-1})/2} & \qquad = \pt*{1 - \frac{1}{2\sqrt{n}-k + 1}}^{(10n^{3/2}\log{n})/2}\\
    & \qquad \leq \pt*{1 - \frac{1}{2\sqrt{n}}}^{5n^{3/2}\log{n}} \\
    & \qquad \leq 2^{-\frac{5}{2}n\log{n}}.
  \end{align*}
  Hence, by \Cref{Prop:ind},
  \begin{align*}
    \max_{\alg : q_{\alg} < q_{n,k}} \pr{\exih_{n,k}} & \qquad \leq 2^{-\frac{5}{2}n\log{n}} + 2^k \cdot 2^{-2n\log{n}+(\sqrt{n}+1)(k-1-\sqrt{n})}\\
    & \qquad \leq 2^{-\frac{5}{2}n\log{n}} + 2^{\sqrt{n}} \cdot 2^{-2n\log{n}+(\sqrt{n}+1)(k-1-\sqrt{n})}\\
    & \qquad \leq 2^{-2n\log{n}+(\sqrt{n}+1)(k-\sqrt{n})}.
  \end{align*}

  The number of vertices in $\treh_{n,k}$ is $d_{n,k}^{\ell_{n,k}} \cdot \prod_{i=k-1}^1 (d_{n,i} - d_{n,i+1}) d_{n,i}^{\ell_{n,i}} \leq  2^{12n^{3}\log^2{n}}$.
\end{proof}

\end{document}